%% file: main.tex
\documentclass[11pt,oneside]{article}
\usepackage[margin=1.2in]{geometry}
\usepackage{amsthm, amsmath, amssymb}

\usepackage[utf8]{inputenc}

\usepackage{microtype}

\usepackage[toc,page]{appendix}
\usepackage{enumitem}

\usepackage{xcolor}
\definecolor{lgray}{gray}{0.6}
\definecolor{dgray}{gray}{0.35}
\usepackage{tikz}
\usetikzlibrary{arrows,shapes,trees, arrows.meta}
\usepackage[gen]{eurosym} 

\newtheorem{theorem}{Theorem}

\newtheorem{lemma}[theorem]{Lemma}
\newtheorem{definition}[theorem]{Definition}

\newtheorem*{theorem*}{Theorem}

\begin{document}

\title{Default Ambiguity: Finding the Best Solution to the \\ Clearing Problem\vspace{25pt}}

\author{\Large Pál András Papp \footnotesize \vspace{7pt} \\ ETH Zürich \vspace{2pt} \\ apapp@ethz.ch
\and \and \Large Roger Wattenhofer \footnotesize \vspace{7pt} \\ ETH Zürich \vspace{2pt} \\ wattenhofer@ethz.ch}

\date{}

\maketitle

\vspace{5pt}

\begin{abstract}
We study financial networks with debt contracts and credit default swaps between specific pairs of banks. Given such a financial system, we want to decide which of the banks are in default, and how much of their liabilities can these defaulting banks pay. There can easily be multiple different solutions to this problem, leading to a situation of \textit{default ambiguity}, and a range of possible solutions to implement for a financial authority.
		
In this paper, we study the properties of the solution space of such financial systems, and analyze a wide range of reasonable objective functions for selecting from the set of solutions. Examples of such objective functions include minimizing the number of defaulting banks, minimizing the amount of unpaid debt, maximizing the number of satisfied banks, and many others. We show that for all of these objectives, it is NP-hard to approximate the optimal solution to an $n^{1-\epsilon}$ factor for any $\epsilon>0$, with $n$ denoting the number of banks. Furthermore, we show that this situation is rather difficult to avoid from a financial regulator's perspective: the same hardness results also hold if we apply strong restrictions on the weights of the debts, the structure of the network, or the amount of funds that banks must possess. However, if we restrict both the network structure and the amount of funds simultaneously, then the solution becomes unique, and it can be found efficiently.
\end{abstract}

\vspace{25pt}

\section{Introduction} \label{sec:Intro}

Financial systems are often called ``highly complex'', suggesting that relations and contracts between different financial institutions such as banks form a networked system that is basically impossible to understand. In order to model this phenomenon, there is a recent line of work that aims to describe this complexity in terms of computational complexity.

At the core of understanding financial systems is the so-called \textit{clearing problem}: given a system of banks with (conditional or unconditional) debt contracts between specific banks, we need to decide which of the banks are in default due to these debts, and how much of their liabilities can these defaulting banks pay. This is a fundamental problem in a financial system, and an essential task for a financial regulator after a shock, with the 2008 financial crisis as a recent example.

Earlier results show that the clearing problem is computationally easy if all contracts between the banks are unconditional debts, or more generally, if the contracts in the network represent \textit{``long'' positions}; that is, a better outcome for one bank ensures a better (or the same) outcome for other banks. However, this is not always the case in practice: banks often have \textit{``short'' positions} on each other, when it is more favorable for a bank if another bank is in a worse situation. Typical short positions are credit default swaps (CDSs), short-selling options and other types of derivatives.

This suggests that a realistic analysis of financial systems requires a model that can capture both long and short positions. However, with both long and short positions in the network, financial systems exhibit significantly richer behavior: we can easily have situations of \textit{default ambiguity} when there are multiple solutions in the system, and none of these solutions is obviously superior to the others in terms of clearing.

In practice, a clearing authority has to make a choice among these different solutions of the system, yielding an outcome that is more favorable to some banks and less favorable to others. In this paper, we focus on such cases of default ambiguity; we study the different solutions of the system, and various criteria that the authority could apply to evaluate these solutions and select one of them to implement.

We begin with some fundamental observations about the solution space of financial systems. We then introduce a wide range of problems that aim to find the best solution in the system according to a specific objective function. These include finding e.g. the solution with the smallest number of defaults, the solution which is preferred by the largest number of banks, the best solution for a specific bank, and many others. 

Our first main contribution is negative, showing that all these problems are not only NP-hard to solve, but also NP-hard to approximate to any $n^{1-\epsilon}$ factor (for any $\epsilon>0$). This shows that even if the clearing authority has a well-defined objective to select among the solutions, finding a reasonably good solution is still not viable in practice.

We then study the same problem from a financial regulator's perspective, showing that it is rather difficult to come up with restrictions on the network to prevent this situation. In particular, we show that the same hardness results still hold in many restricted variants of the model: with unit-weight contracts, with severe restrictions on the network structure, and also if we require banks to own a positive amount of funds.

However, on the positive side, we also show that if we restrict both the network structure and the funds of banks simultaneously, then the resulting financial networks have a unique solution, and this solution can be found efficiently.

\section{Related Work}

The fundamental model of financial systems was introduced by Eisenberg and Noe \cite{model1}, which only assumes simple debt contracts between the banks. Following works have also extended this model by e.g. default costs \cite{model2}, cross-ownership relations \cite{cross2,cross1} or so-called covered CDSs \cite{coveredCDS}. However, these model variants can only describe long positions in a network. This means that there is always a \textit{maximal} solution in the system that is simultaneously the best for all banks, and thus the clearing problem is not particularly interesting in this setting.

In contrast to this, the recent work of Schuldenzucker \textit{et al.} \cite{base1, base2} introduces a model which also allows CDSs in the network, i.e. conditional debt contracts where the payment obligation depends on the default of a specific third bank. While a CDS is still a very simple contract, it already allows us to capture short positions in the network. Moreover, CDSs are a prominent kind of derivative in real-world financial systems that also played a major role in the 2008 financial crisis \cite{CDS3}.

We use this model of Schuldenzucker \textit{et al.} as the base model for our findings. With both debts and CDSs, the clearing problem suddenly becomes significantly more challenging. The work of \cite{base1, base2} mostly focuses on the existence of a solution in this model, and the complexity of finding an arbitrary solution; we summarize these results in Section \ref{sec:basics}.

However, in the general case, these financial networks do not have a maximal solution, and thus an authority has to select from a set of solutions that represent a trade-off between the interests of different banks. The work of \cite{base1, base2} does not study this situation, describing it as unwanted since it is prone to the lobbying activity of banks in the system.
Our work analyzes the clearing problem in this general case; to our knowledge, the problem has not been studied from this perspective before.

In general, there are many previous works that study the propagation of shocks in financial networks, and its dependence on the connectivity of the network \cite{prop1, prop4,prop2, prop6}. There are also several results that study the topic from a computational complexity perspective; however, they mostly assume a simple debt-only model, and focus on more complex questions, such as sensitivity to shocks or bailout policies \cite{Beni,otherComplex, Kleinberg, swapping}. Other works introduce more substantial changes into these models, e.g. time-dependent clearing mechanisms \cite{seqclear2, itcs} or game-theoretic aspects \cite{gametheo, icalp}.

There is also a wide literature on different financial derivatives, and CDSs in particular \cite{CDS4, CDS3, CDS2}. On the more practical side, the clearing problem also plays a central role in stress tests to evaluate the sensitivity of financial systems, e.g. in the European Central Bank's stress test framework \cite{ECB}.

\section{Model definition} \label{sec:model}

\subsection{Banks and contracts}

A financial network consists of a set of \emph{banks} $B$. Individual banks are mostly denoted by $u$, $v$ or $w$, the number of banks by $n=|B|$. Each bank $v$ has a certain amount of \emph{funds} (in financial terms: external assets) available to the bank, denoted by $e_v$.

We assume that there are contracts for payments between given pairs of banks in the system. Each such contract is between two specific banks $u$ and $v$, and obliges $u$ (the debtor) to pay a specific amount of money (known as the \emph{notional}) to the other bank $v$ (the creditor), either unconditionally or based on a specific condition.

These contracts result in a specific amount of payment obligation for each bank $v$. If $v$ cannot fulfill these obligations, then we say that $v$ is \emph{in default}. In this case, the \emph{recovery rate} of $v$, denoted by $r_v$, is the proportion of liabilities that $v$ is able to pay. Note that $r_v \in [0,1]$, and $v$ is in default exactly if $r_v < 1$.

The model allows two kinds of contracts between banks. Debt contracts (or simply \textit{debts}) oblige bank $u$ to pay a specific amount to $v$ unconditionally, i.e. in any case. On the other hand, we also allow \emph{credit default swaps} (\emph{CDSs}) between $u$ and $v$ in reference to a third bank $w$. A CDS represents a conditional debt that obliges $u$ to pay a specific amount to $v$ only in case if bank $w$ is in default. More specifically, if the weight of the CDS is $\delta$ and the recovery rate of bank $w$ is $r_w$, then the CDS incurs a payment obligation of $\delta \cdot (1-r_w)$ from node $u$ to $v$. In practice, CDSs are often used as an insurance policy against the default of the debtors of the bank, or as a speculative bet based on insights into the market.

Before a formal definition, let us consider the example in Figure \ref{fig:example}. In this system, bank $u$ has a total liability of $4$ due to the $2$ outgoing debts, but it only has funds of $2$; hence it is in default, and its recovery rate is $r_u=\frac{2}{4}=\frac{1}{2}$. The model assumes that in this case, it makes payments \emph{proportionally to the respective liabilities} in the contracts; thus it transfers 1 unit of money to $w$ and 1 unit to $v$.

Since $u$ has a recovery rate of $r_u=\frac{1}{2}$, the CDS from $w$ to $v$ translates to a liability of $2 \cdot (1-r_u) = 1$. Although $w$ has no funds, it receives 1 unit of money from $u$, so it can fulfill this payment obligation and narrowly avoids default, $r_w=1$.

Finally, $v$ has no liabilities at all, so $r_v=1$. Since it receives 1 unit of money from both $u$ and $w$, and has $e_v=1$, it has 3 units of money after the clearing of the system.

\subsection{Assets and liabilities}

Formally, our systems are defined by a vector $e=(e_v)_{v \in B}$, the matrix $D=(\delta_{u,v})_{u,v \in B}$, where $\delta_{u,v}$ denotes the weight of debt from $u$ to $v$ (interpreted as $\delta_{u,v}=0$ if there is no such debt), and the matrix $C=(\delta_{u,v}^w)_{u,v,w \in B}$, where $\delta_{u,v}^w$ denotes the weight of the CDS from $u$ to $v$ in reference to $w$. We assume that no bank enters into a contract with itself or in reference to itself. Given a financial system on $B$ by $(e, D, C)$, we are interested in the recovery rates $r_v$ of banks, which can also be represented as a vector $r=(r_v)_{v \in B}$.

Given a recovery rate vector $r$, the \emph{liability} of $u$ to $v$ is formally defined as
\[ l_{u,v}(r)=\delta_{u,v} + \sum_{w \in B} \delta^w_{u,v} \cdot (1-r_w).\]
The total liability of bank $u$ is $l_u(r) = \sum_{v \in B} \; l_{u,v}(r)$, i.e. the sum of payment obligations for $u$. However, the actual payment from $u$ to $v$ can be lower than $ l_{u,v}(r)$ if $u$ is in default. The model assumes that defaulting banks always use all their assets to pay for liabilities, and they make payments proportionally to the respective liabilities. With a recovery rate of $r_u$, $u$ can pay an $r_u$ portion of each liability, so the payment from $u$ to $v$ is $p_{u,v}(r) = r_u \cdot l_{u,v}(r)$.

On the other hand, the \textit{assets} of $v$ are defined as 
\[ a_v(r)=e_v + \sum_{u \in B} \, p_{u,v}(r). \]
Given the assets and liabilities of $v$, the recovery rate $r_v$ has to satisfy $r_v=1$ if $a_v(r) \geq l_v(r)$ (i.e. if $v$ is not in default), and $r_v=\frac{a_v(r)}{l_v(r)}$ if $a_v(r) < l_v(r)$ (if $v$ is in default). If a vector $r$ is an equilibrium point of these equations, i.e. it satisfies this condition on $a_v(r)$ and $l_v(r)$ for every bank $v$, then $r$ is a \emph{clearing vector} of the system. Our main goal is to analyze the different clearing vectors.

The \emph{equity} of $v$ in a solution is defined as 
\[ q_v(r)=\max\left(a_v(r)-l_v(r)^{\,},^{\,} 0\right) ^{\,} , \]
i.e. the amount of money available to $v$ after clearing. In the example of Figure \ref{fig:example}, we have $q_u=0$, $q_w=0$ and $q_v=3$. We assume that the main goal of banks is to maximize their equity. Note that we have written $q_u$ instead of $q_u(r)$ in order to simplify notation; we often do not show the dependence on $r$ when $r$ is clear from the context.

Previous works also consider an extension of this base model with \emph{default costs} \cite{base1, base2, model2}; we also refer to this setting as systems with loss. In this case, the financial network has two more parameters $\alpha, \beta \in [0,1]$, and when a bank goes into default, it loses a specific fraction of its assets. More specifically, if $v$ is in default, then its assets are defined as
\[ a_v(r) = \alpha \cdot e_v + \beta \cdot \sum_{u \in B} p_{u,v}(r) .\]
Throughout the paper, we mostly focus on the base model without loss, i.e. we always assume $\alpha\!=\!\beta\!=\!1$ unless specified otherwise. However, we discuss the extension of our proofs to systems with loss in Appendix \ref{App:B}, and we also briefly study some questions that only arise in case of default costs.

In the rest of the paper, we switch to a computer science terminology: we refer to banks in the system as \textit{nodes}, clearing vectors as \textit{solutions} (with the set of solutions denoted by $S$), and the notionals of contracts as the \textit{weight} of the contracts.

\begin{figure}
\minipage{0.46\textwidth}
\centering
	\vspace{-2pt}
	\input{example.tikz}
	\vspace{-3pt}
	\caption{Example system on $3$ banks. External assets are shown in rectangles besides the bank, simple debts are shown as blue arrows, and CDSs are shown as brown arrows with a dotted line to the reference bank.}
	\label{fig:example}
\endminipage\hfill
\hspace{0.07\textwidth}
\minipage{0.46\textwidth}	
\centering
	\vspace{3pt}
	\input{branch.tikz}
	\vspace{-2pt}
	\caption{Branching gadget consisting of two nodes $x$ and $y$, both having an outgoing debt to a sink $t$ and an incoming CDS from a source $s$.}
	\label{fig:branch}
\endminipage\hfill
\end{figure}

\section{Properties of the solution space} \label{sec:basics}

\paragraph{Previous work.} The work of Schuldenzucker \textit{et al.} mostly focuses on the existence and computability of solutions \cite{base1, base2}. Their main results can be summarized as follows:

\begin{itemize}[topsep=2.5pt, itemsep=3pt, parsep=0pt]
\item \textit{Loss-free systems} ($\alpha=\beta=1$): in this case, there always exists a solution. However, the proof is non-constructive; finding an (approximate) solution is PPAD-hard.
\item \textit{Systems with loss} ($\alpha^{\!}<^{\!}1$ or $\beta^{\!}<^{\!}1$): in this case, a solution might not exist at all. Deciding if a system has an (approximate) solution is an NP-hard problem.
\end{itemize}

Once we know that a solution exists, another natural question is if there exists a \emph{maximal} solution, i.e. a solution $r$ such that $q_v(r) \geq q_v(r')$ for every node $v$ and every solution $r'$. If such a maximal solution exists, then we can assume that a clearing authority always prefers to implement this solution. However, in both settings, a system can easily have multiple solutions with none of them being maximal.

\paragraph{Branching gadget.} A basic building block in our constructions is the \emph{branching gadget} shown in Figure \ref{fig:branch}, which has already been used with some parametrizations in the works of \cite{base1, base2}, e.g. as an example system with no maximal solution. For the weight parameters $\delta_x$ and $\delta_y$ of the gadget, we always assume $\delta_x \geq \delta_y \geq 1$.

Since the source and sink nodes can never go into default, we only analyze the recovery rate subvector $(r_{x}, r_{y})$. First, observe that we cannot have both nodes surviving, i.e. $(1,1)$ as a solution: both nodes only receive any funds if the other node is in default. However, if either $r_x=0$ or $r_y=0$, then the other node can already pay its debt, thus $(0,1)$ and $(1,0)$ are always solutions in this system.

Besides this, there may be other solutions when both nodes are in default with a positive recovery rate; these depend on the concrete values of $\delta_x$ and $\delta_y$. If $x$ and $y$ are in default, then their assets are equal to the amount of debt they can pay, so the remaining solutions are obtained from the equations $r_{x}=\delta_x \cdot (1-r_{y})$ and $r_{y}=\delta_y \cdot (1-r_{x})$.

However, there are also choices of $\delta_x$, $\delta_y$ for which these equations confirm that $(0,1)$ and $(1,0)$ are indeed the only solutions. One such example is $\delta_x=2$, $\delta_y=1$; we refer to this case as the \emph{clean branching gadget}, and we assume this parametrization unless specified otherwise. This gadget variant is a natural candidate for representing a binary choice: $r_{x}$ is either 0 or 1 in any solution, and $r_y$ offers a convenient representation of its negation.

\paragraph{Number of solutions.} Let us now discuss the size of the solution space in our systems.

\begin{lemma}
There exists a financial system with infinitely many solutions.
\end{lemma}

\begin{proof}
Consider the branching gadget of Figure \ref{fig:branch} with $\delta_x=\delta_y=1$. For any $\rho \in [0,1]$, the vector $(\rho, 1\!-\!\rho)$ satisfies the equations above, thus it is a solution of the system.
\end{proof}

While this shows that the number of solutions is potentially unlimited, the difference between most of these vectors is only the extent of the defaults. Thus it is also natural to study another concept of difference between solutions: we say that two solutions $r$ and $r'$ are \emph{essentially different} if there is a node $v$ such that either $r_v=1$ but $r'_v<1$, or $r'_v=1$ but $r_v<1$. Since we only consider a boolean value for each node in this definition, the number of pairwise essentially different solutions is at most $2^n$.

\begin{lemma}
There exists a system with $2^{\Omega(n)}$ solutions that are pairwise essentially different.
\end{lemma}

\begin{proof}
Let us take $\frac{n}{4}$ independent copies of the clean branching gadget. In each gadget, there are two possible subsolutions: (0,1) or (1,0). Over the distinct gadgets, these can be combined in any way, adding up to $2^{n/4}$ solutions that are pairwise essentially different.
\end{proof}

\paragraph{Better and worse solutions.} While financial systems may not always have a maximal solution, it is still reasonable to say that some solutions are better than others.

\begin{definition}
Given two solutions $r$ and $r'$, we say that $r'$ is \emph{strictly better} than $r$ if $q_v(r') \geq q_v(r)$ for every node $v$, and there exists a node $u$ such that $q_u(r')>q_u(r)$. A solution $r$ is \emph{Pareto-optimal} if there is no solution $r'$ that is strictly better than $r$ (otherwise, $r$ is \emph{Pareto-suboptimal}).
\end{definition}

A financial authority might want to avoid implementing Pareto-suboptimal solutions, and prefer a strictly better solution instead. However, selecting among Pareto-optimal solutions is more difficult, since they represent a trade-off between the preferences of different nodes.

We first show that in our base financial system model without loss $(\alpha=\beta=1)$, every solution is Pareto-optimal. One can consider this claim as slight generalization of the similar claim in \cite{model1} for debt-only networks, adapted to our more complex network model.

\begin{lemma} \label{ops:pareto}
In loss-free financial systems, every solution is Pareto-optimal.
\end{lemma}

\begin{proof}
We show that in every solution, $\sum_{v \in B} q_v = \sum_{v \in B} e_v$. Since $\sum_{v \in B} e_v$ is a fixed parameter of the input problem, this already proves the statement.

Recall that in a given solution, $p_v=\sum_{u \in B} p_{v,u}$ and $a_v=e_v + \sum_{u \in B} p_{u,v}$ denotes the payments and assets of node $v$, respectively. Furthermore, the equity of $v$ is always $q_v = a_v - p_v$, regardless of $v$ being in default or not. This implies
\[ \sum_{v \in B} q_v = \sum_{v \in B} (a_v - p_v) = \sum_{v \in B} \left( e_v + \sum_{u \in B} p_{u,v} - \sum_{u \in B} p_{v,u} \right) = \sum_{v \in B} e_v + \sum_{u, v \in B} \left(p_{u,v} - p_{u,v} \right) = \sum_{v \in B} e_v. \]
\end{proof}

However, once we have default costs in the system, then some funds are lost when a node goes into default. Since the total amount of lost funds depends on the number of nodes in default, the funds remaining in the system might differ among different solutions, so some of them might turn out to be Pareto-suboptimal.

\begin{lemma} \label{ops:subopti}
In systems with loss, there can be solutions that are Pareto-suboptimal.
\end{lemma}

\begin{proof}
Let $\beta=\frac{1}{2}$, $\alpha \in (0,1)$, and let us consider the branching gadget with $\delta_x=\delta_y=\frac{3}{2}$. To avoid confusion, let us now assume that $e_{s}=3$ instead of infinity. For simplicity, we express the recovery rate and equity vectors in the gadget by $(r_{s}, r_{x}, r_{y}, r_{t})$ and $(q_{s}, q_{x}, q_{y}, q_{t})$.

The vectors $(1,0,1,1)$ are $(1,1,0,1)$ are still solutions of this system; these induce equity vectors of $(\frac{3}{2}, 0, \frac{1}{2}, 1)$ and $(\frac{3}{2}, \frac{1}{2}, 0, 1)$, respectively. Any other solution must satisfy $r_{x}=\frac{3}{2} \cdot \beta \cdot (1-r_{y})$ and $r_{y}=\frac{3}{2} \cdot \beta \cdot (1-r_{x})$. Solving this system of equations, we get that the third solution is $(1,\frac{3}{7},\frac{3}{7},1)$, resulting in an equity vector of $(\frac{9}{7}, 0, 0, \frac{6}{7})$.

One can observe that this third solution is strictly worse than the previous two solutions.
\end{proof}

\section{Finding the ``best'' solution}

In this section, we discuss a wide range of realistic objective functions for selecting a solution in out networks. We show that for these objectives, the optimal solution is even hard to reasonably approximate. The details of these proofs are discussed in Appendix \ref{App:A}.

\subsection{Tools and gadgets} \label{sec:tools}

We first provide a quick overview of the gadgets that we use as building blocks in our constructions. Note that most of these gadgets have already been used before in the work of \cite{base1, base2}, sometimes in a slightly different form.

Most nodes in these gadget will have the convenient property that their recovery rate is always either $0$ or $1$. Generally, we will say that $v$ is a \emph{binary node} if $r_v \in \{0,1\}$ in any solution.

\begin{itemize}[topsep=2.5pt, itemsep=4pt, parsep=0pt]
 \item \textbf{Clean branching gadget:} this gadget was already discussed in Section \ref{sec:basics}. Recall that the banks $x$ and $y$ represent a binary state: in every solution, we have either $r_x=1$ and $r_y=0$, or $r_x=0$ and $r_y=1$.
 \item \textbf{Cutoff gadget:} given two parameters $0<\eta_1<\eta_2<1$, this gadget from \cite{base2} takes an input node $v$, and transforms it into a binary node $w$ if $r_v \notin (\eta_1, \eta_2)$, ensuring that $r_w=0$ if $r_v \leq \eta_1$, and $r_w=1$ if $r_v \geq \eta_2$. We will only use cutoff gadgets for adapting our results to some restricted model variants.
 \item \textbf{Logical gates:} we can also develop gadgets that simulate boolean operations on nodes. More specifically, given two binary nodes $u$ and $v$, we can construct the following gadgets:
 \begin{itemize}[topsep=3pt, itemsep=2.5pt, parsep=0pt]
	 \item a \textsc{not} gate, i.e. a node $w$ such that $r_{w}=1$ if $r_{v}=0$, and $r_{w}=0$ if $r_{v}=1$,
	 \item an \textsc{or} gate, i.e. a node $w$ such that $r_{w}=0$ if $r_{u}=r_{v}=0$, and $r_{w}=1$ otherwise,
	 \item an \textsc{and} gate, i.e. a node $w$ such that $r_{w}=1$ if $r_{u}=r_{v}=1$, and $r_{w}=0$ otherwise.
 \end{itemize}
\end{itemize}

We demonstrate the \textsc{not} and \textsc{or} gates in Figures \ref{fig:not} and \ref{fig:or}, and discuss the behavior of these gadgets in Appendix \ref{App:0}. Note that Figure \ref{fig:or} already uses the \textsc{not} gate as black box, denoted by a $\neg$ symbol. In a similar fashion, we can also create \textsc{and} and \textsc{or} gates on multiple inputs.

\begin{figure}
\hspace{0.03\textwidth}
\minipage{0.35\textwidth}
\centering
	\vspace{20pt}
	\input{not.tikz}
	\vspace{-13pt}
	\caption{\textsc{not} gate}
	\label{fig:not}
\endminipage\hfill
\hspace{0.12\textwidth}
\minipage{0.4\textwidth}	
\centering
	\input{or.tikz}
	\vspace{-14pt}
	\caption{\textsc{or} gate}
	\label{fig:or}
\endminipage\hfill
\hspace{0.03\textwidth}
\end{figure}

Finally, when adding incoming or outgoing contracts to a bank $v$ in our constructions, our main goal is often to establish a certain behavior for $v$, and thus it is unimportant where these contracts come from or go to. Hence for simplicity, we add a specific \textit{source node} $s$ to our constructions with $e_{s}= \infty$ which is the source of all such incoming contracts, and a specific \textit{sink node} $t$ which is the recipient of all such outgoing contracts.

\subsection{Example: maximizing the equity of a node} \label{sec:equity}

To demonstrate the main idea behind our constructions, we first discuss the problem of maximizing the equity of a specific node. That is, given a node $v$, we define the value of a solution $r$ as the equity $q_v$, and we denote the search problem of finding the highest-value solution by Max\textsc{Equity}$(v)$. This is a very natural problem, and a crucial question for $v$ if it wants to understand its situation in the network.

However, this problem is already hard to solve in our model.

\begin{theorem} \label{th:maxeq}
The problem Max\textsc{Equity}$(v)$ is NP-hard to approximate to any $n^{1-\epsilon}$ factor.
\end{theorem}

\begin{proof}
We use a reduction from the boolean satisfiability (SAT) problem, which is known to be NP-complete \cite{vc}. Given an input boolean formula $\phi$ on $N$ variables and $M$ clauses, we transform this into a financial system representation by creating $N$ distinct branching gadgets, each corresponding to a specific variable. Recall that if we understand $r_x$ to be the value of the variable in an assignment, then $r_y$ represents its negation.

Given these variables, we can use our logical gates to compute the value of $\phi$ for a specific assignment: we first combine each clause into a node with an \textsc{or} gate, and then combine all these nodes with an \textsc{and} gate. This provides a binary indicator node $v_I$ which describes the value of $\phi$ under a specific assignment. We also add a further \textsc{not} gate on top of $v_I$ to obtain a convenient representation of its negation in a new bank $\overline{v_I}$.

Most of our hardness results will use this \textit{base construction}, extended by further gadgets representing the specific objective function. For the Max\textsc{Equity}$(v)$ objective, we only add a node $v$ that has $e_v=0$, and an incoming CDS of weight $n$ in reference to $\overline{v_I}$.

If there exists a satisfying assignment to $\phi$, then there is a solution in this system that has $r_{\overline{v_I}}=0$, and thus $q_v=n$. As such, any $n^{1-\epsilon}$ approximation algorithm must return a solution in this case with $q_v \geq n^{\epsilon} > 0$. On the other hand, if $\phi$ is unsatisfiable, then every solution of the system has $q_v=0$. Hence a polynomial-time approximation would also allow us to decide whether $\phi$ is satisfiable, which completes the reduction. \qedhere
\end{proof}

We point out that the branching gadgets already determine the recovery rate of all other nodes in this system. As such, the system has exactly the $2^N$ solutions that correspond to the different variable assignments. This means that the source of this computational hardness is not the fact that we cannot even find a single solution, as described in \cite{base2} before; in our case, it is not only straightforward to find an (arbitrary) solution in the network, but we can also easily characterize the entire solution space of the system.

We also note that weight of the CDS in the proof was chosen as $n$ in order to demonstrate that inapproximability still holds if we also allow a constant offset besides the $n^{1-\epsilon}$ factor.

With a slightly different gadget appended to the base construction, we can present a similar reduction for the problem of \textit{minimizing} the equity of a bank $v$.

\begin{theorem} \label{th:mineq}
The problem Min\textsc{Equity}$(v)$ is NP-hard to approximate to any $n^{1-\epsilon}$ factor.
\end{theorem}

\begin{proof}
This only requires a slight modification to the same setting: we now need to add a bank $v$ with $e_v=n$, and an outgoing CDS of weight $n$ in reference to $\overline{v_I}$. With this the optimum value is $q_v=0$ if $\phi$ is satisfiable, and $q_v=n$ otherwise. \qedhere
\end{proof}

\subsection{Global objective functions} \label{sec:globopt}

Given a financial system with many solutions, there are various objectives that an authority could follow when choosing the solution to implement. Some of the most natural objective functions are as follows:

\begin{itemize}[topsep=3pt, itemsep=5.5pt, parsep=0pt]
 \item Min\textsc{Default}: minimize the number of defaulting nodes, i.e. minimize $|\{ v \in B \, | \, r_v<1 \}|$
 \item Max\textsc{Prefer}: find the solution that is the primary preference of most nodes, i.e. define the maximal equity of bank $v$ as $q_v\,\!\!^{(\text{max})} = \text{max}_{r \in S} \; q_v(r)$, and then maximize $|\{ v \in B \, | \, q_v(r)=q_v\,\!\!^{(\text{max})} \}|$,
 \item Min\textsc{Unpaid}: minimize the amount of unpaid liabilities, i.e. minimize $\sum_{u,v \in B} l_{u,v} - p_{u,v}$. 
\end{itemize}

One can show that these are indeed different problems: they can obtain their optimum in distinct solutions, and the optimum for one objective might give a very low-quality solution in terms of another one.

\begin{theorem} \label{obs:twoOpt}
For any objectives $f_1$, $f_2$ from above, there is a system such that in the optimal solution for $f_1$, the value of $f_2$ is a $\Theta(n)$ factor worse than the optimum value of $f_2$.
\end{theorem}

\noindent We provide example constructions for these claims in Appendix \ref{App:DiffOpt}. In fact, one can even combine these examples into a single system with a very different optimum for each function. 

\begin{theorem} \label{obs:diffOpt}
There exists a financial system such that the optima for the objective functions above are all obtained in different solutions, and in terms of the respective metrics, each of these optima are a factor of $\Omega(\sqrt{n})$ better than any other solution in the system.
\end{theorem}

Now let us analyze these problems from a complexity perspective. We can apply a similar technique to Theorem \ref{th:maxeq} to show that it is hard to approximate any of these objectives.

\begin{theorem} \label{th:mindef}
The problem Min\textsc{Default} is NP-hard to approximate to any $n^{1-\epsilon}$ factor.
\end{theorem}

\renewcommand{\proofname}{Proof sketch.}

\begin{proof}
Given a fixed constant $\epsilon$, let us select an $\epsilon'$ such that $0<\epsilon'<\epsilon$. Also, given a formula $\phi$ on $N$ variables and $M$ clauses, let us introduce $m:=N+M$. We extend the base construction of Section \ref{sec:equity} by introducing $m^{1/\epsilon'}$ distinct new banks $u_i$ to the system that all have $e_{u_i}=0$, and an outgoing CDS of weight $1$ in reference to the indicator node $v_I$.

For every variable assignment that evaluates to false, we have $r_{v_I}=0$, so all the new nodes are in default; as such, the number of defaulting nodes is $m^{1/\epsilon'} + O(m)$. On the other hand, if there is a satisfying assignment, then the banks $u_i$ have no liability in the corresponding solution, so the number of defaulting banks is only $O(m)$. Since $n=\Theta(m^{1/\epsilon'})$ in this system, the best solution has either $\Theta(n)$ or $O(n^{\epsilon'})<n^{\epsilon}$ defaults, depending on whether $\phi$ is satisfiable; this shows an inapproximability to any $n^{1-\epsilon}$ factor.

Since $\epsilon'$ is a constant, our construction on $O(m^{1/\epsilon'})$ nodes still has a size that is polynomial in the size $m$ of the original formula $\phi$. As such, any polynomial-time approximation algorithm in $n$ would also have a running time that is polynomial in $m$.\qedhere
\end{proof}

We can also rephrase the Min\textsc{Default} problem as maximizing the number of surviving (non-defaulting) nodes; the two problems clearly have the same optimal solution.  However, this Max\textsc{Surviving} problem is defined by a different metric in its objective function, so it could behave very differently in terms of approximability (see e.g. the minimum vertex cover and maximum independent set problems, which are also complements \cite{vc,inset}). However, it turns out that in our case, the problem is hard to approximate in both metrics.

\begin{theorem} \label{th:maxsur}
The problem Max\textsc{Surviving} is NP-hard to approximate to any $n^{1-\epsilon}$ factor.
\end{theorem}

We can use different variants of the same proof technique to show the same hardness result for the other two objectives. Furthermore, similar to Max\textsc{Surviving}, we can also define dual problems for these objectives, which are also hard to approximate.

\begin{theorem} \label{th:prefer}
The problems Max\textsc{Prefer} and Min\textsc{Unpaid} (as well as their dual problems Min\textsc{LeastPrefer} and Max\textsc{Paid}) are NP-hard to approximate to any $n^{1-\epsilon}$ factor.
\end{theorem}

\subsection{More complex objectives}

\paragraph{Most balanced solution.} In a slightly different setting, an authority could want to find a solution where the distribution of equity is balanced in some sense. E.g. if we have two larger alliances of banks (i.e. sets of nodes), then our goal might be to find a solution that distributes the total equity evenly between these alliances.

In the simplest case of this problem, we consider two nodes $v_1$ and $v_2$, and we define the problem Min\textsc{Diff}$(v_1,v_2)$ of finding the solution where $|q_{v_1}-q_{v_2}|$ is minimal.

\renewcommand{\proofname}{Proof.}

\begin{theorem} \label{th:balance}
The problem Min\textsc{Diff}$(v_1,v_2)$ is NP-hard to approximate to any $n^{1-\epsilon}$ factor.
\end{theorem}

\begin{proof}
We can simply consider the Min\textsc{Equity}$(v)$ construction with $v_1:=v$, and add an extra bank $v_2$ such that $q_{v_2}=0$. This system has $|q_{v_1}-q_{v_2}| = q_{v_1}$, so we can apply the same reduction as in the Min\textsc{Equity} case. \qedhere
\end{proof}

\noindent This already shows that the more general problem of minimizing $|\sum_{v_1 \in V_1} \, q_{v_1} - \sum_{v_2 \in V_2} \, q_{v_2}|$ for two sets of nodes $V_1$ and $V_2$ is also hard. One can also show that the problem still remains hard in the special case when $V_1 \cup V_2 = B$, i.e. if the alliances cover the whole system.

\paragraph{Most representative solution.} It could also be a reasonable goal to select a solution that is somehow representative of the whole solution space $S$. Assuming a fixed distance metric between two solutions (for example, let $d(r, r'):=\sum_{v \in B} |r_v-r_v'|$), there are many natural ways to define a metric of centrality for a given solution $r$ in $S$.

We only discuss one natural approach here: let us define the centrality of a solution $r$ as
\[ cent(r)=\frac{1}{|S|}\sum_{r' \in S} \, d(r, r') \, , \]
and let Min\textsc{Dist} denote the problem of finding the solution $r$ with the lowest $cent(r)$ value.

Note that our result essentially shows that the solution space can exhibit a threshold behavior between two very different shapes, and it is already hard to decide which of the two shapes is obtained. This suggests that the problem is also hard in any other reasonable formulation, i.e. for other distance functions or centrality metrics.
 
\begin{theorem} \label{th:repr}
The problem $Min\textsc{Dist}$ is NP-hard to approximate to any $n^{1-\epsilon}$ factor.
\end{theorem}

\renewcommand{\proofname}{Proof sketch.}

\begin{proof}

The main idea is to add two large sets of nodes to our construction, as sketched in Figure \ref{fig:repr}. The \emph{generating group} consists of $N^2$ independent branching gadgets, while the \emph{control group} has $m^{1/\epsilon'}$ single nodes with an outgoing debt (where $m$ denotes the size of $\phi$ and $\epsilon'<\epsilon$ as before). We ensure that both groups only receive funds if $r_{v_I\!}=1$; otherwise, all the new nodes are in default.

Since the control group contains almost all of the nodes asymptotically, the centrality of a solution is essentially defined by the recovery rates of the nodes in the control group. If $\phi$ is unsatisfiable, then every assignment produces $r_{v_I}=0$, and thus the control nodes have recovery rates of $0$ in every solution. On the other hand, if $\phi$ is satisfiable, then the branching gadgets in the generating group will introduce $2^{N^2}$ new solutions (for each satisfying assignment), which reduces the at most $2^N$ unsatisfying solutions to an asymptotically irrelevant part of $S$. In this case, the control nodes have a recovery rate of $1$ in almost every solution.

Hence the two cases are very different in terms of solution space. An approximation algorithm would always need to find a satisfying assignment if one exists; otherwise, it returns a solution with an average distance of at least $m^{1/\epsilon'\!} \approx n$, while the optimum has a distance of only $O(m) \approx n^{1-\epsilon'}$. \qedhere
\end{proof}

\begin{figure}
\minipage{0.56\textwidth}
	\centering
	\vspace{6pt}
	\resizebox{1.02\textwidth}{!}{\input{mostRepr.tikz}}
	\caption{Construction of Theorem \ref{th:repr}}
	\label{fig:repr}
\endminipage\hfill
\hspace{0.02\textwidth}
\minipage{0.41\textwidth}	
	\vspace{10pt}
	\resizebox{0.98\textwidth}{!}{\input{subopt.tikz}}
	\vspace{5pt}
	\caption{Construction of Theorem \ref{th:better}}
	\label{fig:subopt}
\endminipage\hfill
\end{figure}

\paragraph{Strictly better solution.} \label{sec:better}

Recall that in case of systems with loss, we can also have Pareto-suboptimal solutions, so it is natural to ask if a specific solution can be improved: if there is a solution $r'$ strictly better than $r$, then we would probably want to implement $r'$ instead of $r$. If such an $r'$ was easy to find, then we could iteratively improve an initial solution until we eventually find a Pareto-optimal solution.

\begin{theorem} \label{th:better}
Given a solution $r$, it is NP-hard to decide if $r$ is Pareto-suboptimal.
\end{theorem}

\begin{proof}
The construction, shown in Figure \ref{fig:subopt}, is built around a binary node $v_0$. To each node $u$ of our base construction, we add a so-called \emph{unhappy penalty gadget}. This essentially means that if $r_{v_0}=0$, then $u$ pays a large penalty to a special sink $t_0$; however, $t_0$ has further gadgets attached to ensure that $t_0$ is still worse off if $r_{v_0}=0$, even though it receives money from this penalty. As such, the default of $v_0$ is not favorable to any node in the system; note that this is only possible in systems with loss.

The base idea then is to add another node $w$, which, on the other hand, receives $1$ unit of money if either $r_{v_0}=0$, or $r_{v_I}=1$. Let $r$ be the solution where $r_{v_0}=0$, and thus all nodes in the base construction are in default, but $q_w=1$. Any solution strictly better than $r$ must also have $q_w \geq 1$. If $v_0$ is not in default, this is only possible if we find a satisfying assignment of $\phi$, thus ensuring $r_{v_I}=1$. \qedhere
\end{proof}

\section{Restricted financial networks} \label{sec:models}

Our final goal in the paper is to understand the key reasons behind this computational complexity, and whether we can introduce some restrictions to our network model to eliminate this phenomenon. In particular, we show that the same hardness results also hold in many severely restricted variants of our financial system model, and it takes a combination of multiple restrictions to ensure that the solution space is sufficiently simple. These model variants and the corresponding proofs are discussed in more detail in Appendix \ref{App:B}.

Before considering restrictions to the network, let us first briefly discuss a familiar extension of the model: default costs. We point out that while our hardness results were mostly shown for systems without loss, they can also be extended to systems with loss with some minor modifications.

\begin{theorem} \label{th:withloss}
Theorems \ref{th:maxeq}--\ref{th:mineq} and \ref{th:mindef}--\ref{th:better} also hold for any $\alpha, \beta \in (0,1]$. 
\end{theorem}

\subsection{Unweighted networks}

For convenience, we have sometimes used rather large edge weights in our constructions. One could argue that this is unrealistic, since in practice, the payment obligations are often in the same magnitude. As such, we first show that our hardness results also carry over to the setting when each contract in the network has the same weight.

\begin{theorem} \label{th:unweight}
Theorems \ref{th:maxeq}--\ref{th:mineq} and \ref{th:mindef}--\ref{th:better} also hold in unit-weight networks. 
\end{theorem}

\renewcommand{\proofname}{Proof sketch.}

\begin{proof}
The modifications required for this setting are rather straightforward: most edges in our constructions have unit weight to begin with. Whenever the weight is a larger integer $k$, we can usually split this into $k$ distinct contracts that come from/go to $k$ distinct source/sink nodes. The only cases that require some extra consideration are the gadgets used in Theorems \ref{th:repr} and \ref{th:better}.
\end{proof}

\subsection{Restricted network structure}

In their work, Schuldenzucker \textit{et al.} also discuss several restrictions to the network structure \cite{base1, base2}. While they study these restrictions from a different perspective (their goal is to ensure that the system always has a solution, even with default costs), it is natural to ask whether our hardness results still hold in these restricted network models.

In particular, the authors define the so-called \emph{dependency graph} to express the relations of banks in a directed graph with edges of two colors:
\begin{itemize}[topsep=5pt, itemsep=0pt, parsep=6pt]
\item \textit{Green edges}: intuitively, these indicate long positions. For example, there is a green edge from $u$ to $v$ if $u$ has a contract towards $v$ (debt or CDS), or if $v$ has an outgoing CDS in reference to $u$.
\item \textit{Red edges}: intuitively, these indicate short positions. There is a red edge from $w$ to $v$ if $v$ has an incoming CDS in reference to $w$ (unless there is a debt of even larger weight from $w$ to $v$).
\end{itemize}
For details on the dependency graph, we refer the reader to Appendix \ref{App:B} or the work of \cite{base1}.

The work of \cite{base1} studies different restrictions to the network based on this dependency graph. In the most restricted case, they study systems where the dependency graph contains exclusively (or almost exclusively) green edges, so short positions are essentially banned.

\begin{definition}
We say that a financial network is a \emph{green system} if its dependency graph only contains green edges.
\end{definition}

Using a fixed-point theorem, one can show that green systems are similar to debt-only networks in the sense that they always contain a maximal solution. As such, this simpler case is not so interesting for us in terms of default ambiguity.

On the other hand, \cite{base1} also studies a more general setting where short positions are still allowed in the network, but only in a structurally restricted fashion.

\begin{definition}
A financial network is an \emph{RFC (red-free cycle) system} if no directed cycle of the dependency graph contains a red edge.
\end{definition}

The authors show that in RFC systems, one can always find a solution efficiently. Intuitively, one can iterate through the strongly connected components (SCCs) of the dependency graph in topological order, since every SCC is only dependent on the preceding ones. Since each SCC is a green system, there is always a maximal subsolution in the current SCC (if the subsolutions in previous SCCs are already fixed), and we can find this efficiently.

In contrast to this, our goal of finding the best solution is still not straightforward in these RFC systems. In particular, selecting a different (non-maximal) solution in the first SCC could allow us to find a different solution in the second SCC; while this is unfavorable to banks in the first SCC, it might be much better in terms of our global objective. In fact, our hardness results even hold in this heavily restricted class of networks.

\begin{theorem} \label{th:rfc}
Theorems \ref{th:maxeq}--\ref{th:mineq} and \ref{th:mindef}--\ref{th:better} also hold in RFC systems. 
\end{theorem}

\begin{proof}
The key observation is that directed cycles are in fact very rare in the dependency
graphs of our constructions: we mostly use logical gates that follow a specific ordering, and thus the dependency graphs are already very close to DAGs. The only exception is within the branching gadgets, where banks $x$ and $y$ have short position on each other, and hence there is a red edge between them in both directions. As such, it is sufficient to come up with an alternative branching gadget design that satisfies the RFC property.

The main idea of this gadget is to consider two banks $v_1$ and $v_2$ as in Figure \ref{fig:cyc}. For any $\rho \in [0,1]$, $r_{v_1}=r_{v_2}=\rho$ is a solution of this system.

We can then use the small and large $\rho$ values in this system to represent the two binary states; this can be achieved by creating two banks $x$ and $y$ as the outputs of two cutoff gadgets on $u$, having parameters e.g. $\eta_1=\frac{1}{3}, \eta_2=\frac{1}{2}$ and $\eta_1=\frac{1}{2}, \eta_2=\frac{2}{3}$, respectively.

Finally, we exclude the intermediate $\rho$ values by appending further gadgets to artificially make the solution significantly worse (in terms of our desired objective function) whenever we have $\rho \in [\frac{1}{3},\frac{2}{3}]$. This means that in any reasonable solution, we will have either $\rho \leq \frac{1}{3}$ or $\rho \geq \frac{2}{3}$, and hence either $r_x=1, r_y=0$ or $r_x=0, r_y=1$.
\end{proof}

\begin{figure}
\centering
	\input{cycle.tikz}
	\caption{A simple debt-only network with multiple solutions}
	\label{fig:cyc}
\end{figure}
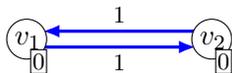

We note that the situation in Figure \ref{fig:cyc} seems rather artificial. However, recall that default ambiguity often arises after an external shock hits the market; as such, one should imagine this as a situation where banks in a cycle have lost all their funds due to such an event.

\subsection{Green systems and regularity}

Our alternative construction in Theorem \ref{th:rfc} uses the fact that a debt-only network can still have multiple solutions in some special edge cases, thus allowing us to create a large solution space. To prevent this phenomenon, we first need a deeper understanding of these cases when green system have multiple solutions.

The work of \cite{model1} already studies this question in debt-only networks, showing that the solution is unique if from any bank there is a directed path to another bank with positive funds. We prove a more general version of this result, extending the theorem to any green system, and using a weaker assumption on the topology. In particular, we show that green systems can only have multiple solutions in a special edge case: when we have a network segment with no funds and no incoming assets at all.

\begin{theorem} \label{th:ambiguity}
Let $G$ be a green system, and assume that $v$ is a bank that has two distinct recovery rates $r_v \neq r_v'$ in two solutions. Let $C$ be the set of nodes reachable from $v$ on a path of simple debts. Then the following must hold:
\begin{itemize}[topsep=3pt, itemsep=3pt, parsep=0pt]
\item for all $u \in C$ we have $e_u=0$,
\item if there is a path of contracts from a bank $w \in G$ to a bank $u \in C$, then $e_w=0$.
\end{itemize}
\end{theorem}

\begin{proof}
The main steps of the proof are as follows:
\begin{itemize}[topsep=4pt, itemsep=1pt, parsep=2pt]
{\setlength\itemindent{-2pt}
 \item Recall from before that a green system always has a maximal solution $r$ (and also a minimal solution $r'$); these assigns the highest/lowest recovery rate to all banks.
 \item In such a setting, all banks must have the same equity in any solution. Intuitively, in systems without loss, if a bank had less equity in a solution $r_0$ than in the maximal solution $r$, then some other bank would need to have more equity in $r_0$ than in $r$.
 \item If $r_{v}>r_{v}'$ (i.e. $v$ can have different recovery rates), then $v$ makes strictly more payment on its outgoing debts in $r$ than in $r'$. In a loss-free system, these extra payments traverse the network in until they either (i) reach a node $u$ with no more unfulfilled liabilities, or (ii) they arrive back at $v$. However, the first option is not possible, since this would mean $q_u>q_u'$; hence all such payments must ultimately arrive back at $v$.
 \item This means that from $v$, any path of contracts (with positive liability) must eventually lead back to $v$, implying that these contracts form an SCC $C$.
 \item Finally, no node $u \in C$ can have $e_u>0$, and also no node $w \in G$ can have a positive payment towards a bank $u \in C$. This is because $C$ is closed under outgoing payments, so if any funds arrive in $C$, then the loss-free property implies that some node $u \in C$ must have $q_u'>0$. This means that we already have $r_u'=1$ in the solution $r'$. However, if $r_v>r_v'$, then in $r$ there is a strictly positive extra payment arriving at $u$; this implies $q_u>q_u'$, which is again a contradiction.
} \qedhere
\end{itemize} 
\end{proof}

Note that the proof also makes a structural observation that the banks reachable from $v$ must form a SCC in the graph of ``meaningful'' contracts (which induce a positive liability in some solution). However, since it is not immediately clear whether a CDS is meaningful, we expressed Theorem \ref{th:ambiguity} in a weaker form, stating the restrictions only for the set of nodes $C$ that are reachable from $v$ on simple debts.

The situation described in Theorem \ref{th:ambiguity} is a very special case, so there are various ways to ensure that we exclude such networks. One natural approach is to restrict the amount of funds that banks must possess, since this is usually strictly supervised in practice.

\begin{definition}
We say that a financial system $G$ is \emph{regular} if we have $e_v>0$ for all $v \in B$.
\end{definition}

This assumption is realistic in many legal frameworks: financial regulations usually require banks to possess enough funds to cover at least a specific portion of their liabilities. Considering that default ambiguity often happens after a shock hits the market, an alternative (more practical) interpretation of this property is that all banks must keep at least some of their funds in a format that is resilient to external shocks.

Note that there are various other options to exclude the edge case of Theorem \ref{th:ambiguity} with weaker conditions; however, most of these approaches are difficult to enforce from a regulator's perspective.

On the other hand, note that Theorem \ref{th:ambiguity} only applies to green systems. If our network is not a green system, then even this rather strong condition is not sufficient to ensure that the solution is unique.

\begin{theorem} \label{th:reg}
Theorems \ref{th:maxeq}--\ref{th:mineq} and \ref{th:mindef}--\ref{th:better} also hold in regular financial systems. 
\end{theorem}

\begin{proof}
The main idea is to consider a new representation of the binary states in our gadgets: instead of $r_v=0$ and $r_v=1$, the two binary states will be represented by $r_v=0.5$ and $r_v=1$. This allows us to give some funds to every node in our construction, thus fulfilling the regularity condition.

Most of our gadgets are actually rather easy to adapt to this setting; it is again only the constructions of Theorems \ref{th:repr} and \ref{th:better} where this is somewhat more technical.
\end{proof}

\subsection{Combined restrictions: a unique solution}

This shows that we need both the RFC property and regularity together to ensure that the solution of the system is unique, and thus our hardness results can be avoided. This provides an interesting final message from our analysis: it suggests that financial regulators might need to use both topological and fund-based restrictions simultaneously in order to eliminate the computational problems arising from default ambiguity.

\begin{theorem} \label{th:unique}
If a system is both regular and RFC, then it has a unique solution. This solution can be efficiently approximated in polynomial time.
\end{theorem}

\renewcommand{\proofname}{Proof.}

\begin{proof}
We can now apply the approach of \cite{base1} for RFC systems, computing a solution by visiting the SCCs in topological order. The payments coming from the previous SCCs can simply be considered as extra funds at the bank when processing the current SCC of the network.

Due to the RFC property, the current SCC is always a green system. Regularity implies that every node $u$ in the SCC has $e_u>0$; this is only further increased by the payments from previous SCCs. As such, Theorem \ref{th:ambiguity} shows that there is always a unique subsolution in the current SCC. Altogether, this implies that the solution $r$ is unique in the whole network; as such, we can indeed simply apply the algorithm of \cite{base1} for RFC systems, which always finds an arbitrary solution.

Note, however, that the solution of our networks can also be irrational in some cases, so we can only claim that it is efficiently approximated with this method. It is already discussed in \cite{base1, base2} that given an error margin $\epsilon>0$, this algorithm finds a recovery rate vector $r^{\epsilon}$ such that $|r_v-r^{\epsilon}_v| \leq \epsilon$ for all $v \in B$, and its running time is polynomial in $n$ and $1/\epsilon$.
\end{proof}

Finally, we point out that if we have default costs, then our hardness results still hold even in the setting of Theorem \ref{th:unique}. This is because with default costs, a green system can still have multiple solutions even if it is regular. If we modify Figure \ref{fig:cyc} to have $e_u=e_v=\frac{1}{3}$ and we assume $\alpha=\beta=\frac{1}{2}$, then both $r_u=r_v=1$ and $r_u=r_v=\frac{1}{3}$ are solutions; while the former is clearly better for $u$ and $v$, the latter might be superior in terms of our objective.

\subsection*{Acknowledgements}

We would like to thank Steffen Schuldenzucker for his very valuable feedback and improvement ideas.

\bibliographystyle{splncs04}
\bibliography{references}

\newpage
\begin{appendices}

\section{Gadgets and logical gates} \label{App:0}

This section provides a more detailed overview of the basic gadgets we apply in our constructions. Note that we have already discussed the clean branching gadget in the main part of the paper. For details on the cutoff gadget, we refer the reader to the work of \cite{base2}.

It remains to discuss the logical \textsc{not}, \textsc{or} and \textsc{and} gates. We again point out that some of these gadgets (in particular, the \textsc{not} gate, and a similar gadget which behaves as a \textsc{nand} gate) were also used before in \cite{base1, base2}. Figure \ref{fig:notation} summarizes the notation of these gadgets in our figures.

For negation, we can consider the simple gadget in Figure \ref{fig:not}. If $r_{v}=0$, then $w$ receives 1 unit of money, and it can pay its debt entirely. However, if $r_{v}=1$, then $w$ has no assets at all, and thus $r_{w}=0$.

The gadget for the \textsc{or} relation is shown in Figure \ref{fig:or}; note that it already uses the previously described \textsc{not} gadget. If $r_{u}$ or $r_{v}$ is 1, then at least one of the connected \textsc{not} gadgets is in default, and thus $w$ has assets of at least 1; this already implies $r_{w}=1$. Otherwise $w$ has no assets at all, and hence we have $r_{w}=0$.

Finally, one possibility to implement the \textsc{and} relation is illustrated in Figure \ref{fig:and}. In this case, if at most one of the nodes $r_{u}$ and $r_{v}$ is 1, then $w_0$ receives a payment on at most one of the two CDSs. Therefore, $w_0$ has at most 1 assets, thus $r_{w_0} \leq \frac{1}{2}$. Since the connected cutoff gadget has $\eta_1=0.7$, we have $r_{w}=0$ in this case. On the other hand, if $r_{u}=r_{v}=1$, then $r_{w_0}=1$, and $r_{w}=1$ follows.

We point out, however, that this version of the \textsc{and} gate is more difficult to adapt to different variants of the network model, so it is often a more convenient solution to express the \textsc{and} relations with a combination of \textsc{not} and \textsc{or} gates instead.

Note that all of these gadgets only use the input nodes $v_1$ and $v_2$ as reference entities for CDSs, and thus inserting such gadgets has no effect on the behavior of the input nodes.

\section{Details of the hardness proofs} \label{App:A}

In this section we discuss our inapproximability proofs in more detail. Note that all of these proofs begin with the use of the base SAT construction, and then they append further gadgets on the indicator nodes $v_I$ and $\overline{v_I}$ to express a specific objective function.

The construction for Max\textsc{Equity}$(v)$ has already been described in Section \ref{sec:equity}. We note here that the idea of this proof suggests that we could obtain similar hardness results for even higher factors than $n$ if we use e.g. very large edge weights. However, since this makes the model somewhat unrealistic, we limit our interest to approximations of up to a factor $n$.

\begin{figure}
\minipage{0.53\textwidth}
\centering
	\vspace{42pt}
	\resizebox{1.0\textwidth}{!}{\input{notation.tikz}}
	\vspace{-13pt}
	\caption{Notation of our gadgets}
	\label{fig:notation}
\endminipage\hfill
\hspace{0.05\textwidth}
\minipage{0.41\textwidth}	
\centering
	\resizebox{1.0\textwidth}{!}{\input{and_alter.tikz}}
	\vspace{-15pt}
	\caption{\textsc{and} gate}
	\label{fig:and}
\endminipage\hfill
\end{figure}

\subsection{Global objectives in Section \ref{sec:globopt}}

Theorem \ref{th:maxsur} can be shown with the same construction as in Theorem \ref{th:mindef}. For every assignment where $\phi$ evaluates to false, the new nodes are all in default, so the number of surviving nodes is only $O(m)$. On the other hand, for a satisfying assignment, the number of surviving nodes is at least $m^{1/\epsilon'}$. This again creates a factor of $n^{1-\epsilon'}$ difference between the two cases, so any approximation algorithm must return a solution with at least $\omega(m)$ surviving nodes if $\phi$ is satisfiable; this completes our reduction.

For Max\textsc{Prefer}, we modify this construction by setting $e_{u_i}=1$ at the extra nodes. As such, a satisfying assignment ensures that these nodes all have $q_{u_i}=1$, while an unsatisfying assignment implies $q_{u_i}=0$. Hence if $\phi$ is satisfiable, then the satisfying assignment is the primary preference of $m^{1/\epsilon'}$ nodes, while the remaining solutions are preferred by at most $O(m)$ nodes; as such, in order to give an $n^{1-\epsilon'}$ approximation, an algorithm would have to find a satisfying assignment in polynomial time.

In case of the dual problem Min\textsc{LeastPrefer} where we define $q_v\,\!\!^{(\text{min})} = \text{min}_{r \in S} \; q_v(r)$ and minimize $|\{ v \in B \, | \, q_v(r)=q_v\,\!\!^{(\text{min})} \}|$, we can use the same construction: if we have a satisfying assignment, then such an assignment is the least preferred solution to at most $O(m)$ nodes, while an unsatisfying assignment is the least preferred solution to $m^{1/\epsilon'}$ nodes. As such, any approximation algorithm needs to find a satisfying assignment if one exists.

For Min\textsc{Unpaid}, it once again suffices to use the construction of Theorem \ref{th:mindef}. For any unsatisfying assignment, the nodes $u_i$ create a total unpaid debt of $m^{1/\epsilon'}$ in the system, besides the unpaid debts in the base construction. On the other hand, with $r_{v_I}=1$, the amount of unpaid debt is only $O(m)$ altogether.

For the dual problem of maximizing $\sum_{u,v \in B} \, p_{u,v}$, we can slightly change this construction: we set $e_{u_i}=1$, and change the reference nodes of the outgoing CDSs to $\overline{v_I}$. This way, the extra nodes can all pay their liabilities in case of $r_{v_I}=1$, so a true assignment results in a paid debt of $m^{1/\epsilon'}$. On the other hand, any false assignment only has a paid debt of $O(m)$ altogether.

\subsection{Most balanced solution}

We have already seen that the setting of Theorem \ref{th:balance} is rather easy to reduce to the case of Min\textsc{Equity}$(v)$. This also settles the general case of minimizing the equity difference between two subsets of nodes $V_1$ and $V_2$.

In a slight detour, we now also briefly discuss another interesting special case of this general setting: what if $V_1$ and $V_2$ form a disjoint partitioning of the whole node set $B$, i.e. the alliances cover the whole system?

For this case, we adapt a similar approach to the Min\textsc{Default} construction; however, we now add $m^{1/\epsilon'}$ distinct sink nodes $t_i$ to the construction (note that strictly speaking, this is not a necessary modification for our proof, but it provides a more realistic construction that does not require a very large amount of funds at a single node). Then for each $i \in [1, m^{1/\epsilon'}]$ we set $e_{u_i}=e_{t_i}=1$, and we create a CDS of weight $1$ from $u_i$ to $t_i$ which is in reference to the indicator node $v_I$.

Now let $V_1$ contain all the nodes $u_i$ in our network, and $V_2$ consist of all the remaining nodes (i.e. the sinks $t_i$ and the nodes of the base construction). If we consider a satisfying assignment in this network, then $r_{v_I}=1$, and thus there is no liability between the newly added banks. This implies that the total equity in $V_1$ is $m^{1/\epsilon'}$, while the total equity in $V_2$ is $m^{1/\epsilon'}+O(m)$. This amounts to a difference of only $O(m)$ between the two sets.

On the other hand, any unsatisfying assignment implies that $V_1$ will have no equity at all, while $V_2$ still has an equity of more than $m^{1/\epsilon'}+O(m)$; hence the difference in this case is at least $m^{1/\epsilon'}$. As such, any approximation algorithm needs to find a satisfying assignment.

\subsection{Most representative solution} \label{App:repr}

We continue with the problem of finding the most representative solution.

\renewcommand{\proofname}{Proof of Theorem \ref{th:repr}.}

\begin{proof}
As outlined before, we add two large sets of nodes to the base construction: the generating group and the control group. The generating group consists of $N^2$ distinct branching gadgets, with the source nodes of these gadgets replaced by a common new node $s_g$. Let us now slightly change our previous notation, and use $m:=\max(N^2, M)$; i.e. $m$ is selected such that the base construction and the generating group altogether contains only $O(m)$ nodes.

Given a constant $\epsilon>0$, we again select a smaller constant $\epsilon' \in (0, \epsilon)$. Then we set the control group to consists of $m^{1/\epsilon'}$ distinct nodes $u_i$, each having $e_{u_i}=0$, a debt of weight 1 towards $t$, and an incoming debt of $1$ from a new common node $s_c$. The nodes $s_g$ and $s_c$ have no funds, but we add a CDS of weight $\infty$ from $s$ to both $s_g$ and $s_c$ in reference to $\overline{v_I}$.

Note that we have only chosen to use the two pseudo-source nodes $s_g$ and $s_c$ to allow a cleaner illustration in Figure \ref{fig:repr}. Instead, it would also be possible to introduce a separate source node with funds of $3$ for each branching gadget, and a separate source with funds of 1 for each $u_i$, and make the payments to each branching gadget/control group node based on a separate CDS in reference to $\overline{v_I}$. This change does not affect our distance metrics since these sources always have a recovery rate of $1$; furthermore, executing the change is indeed necessary if we want to adapt our setting to the case of unit-weight contracts.

The main idea is that in any solution that does not satisfy $\phi$, we have $r_{s_g}=r_{s_c}=0$. In the generating group, this implies that none of the branching gadgets have any assets, and thus all nodes in these gadgets are in complete default (i.e. have recovery rates of $0$); with $r_{v_I}=0$, this is the only subsolution of this subsystem. In the control group, this means that all the nodes $u_i$ are in complete default, too.

However, if there is a satisfying assignment, then this gives infinitely many assets to both $s_g$ and $s_c$. Hence each branching gadget in the generating group indeed offers a binary choice, thus introducing $2^{N^2}$ distinct solutions for each satisfying assignment. In all of these solutions, the nodes $u_i$ in the control group all have $r_{u_i}=1$. Thus if we have at least one satisfying assignment, then the number of solutions with $r_{u_i}=0$ becomes asymptotically irrelevant.

More specifically, assume that $\phi$ has a satisfying assignment, and let us show that an approximation algorithm for Min\textsc{Dist} must return a solution corresponding to a satisfying assignment in this case. For simplicity, let us first assume that there is only one satisfying assignment.

If $r_1$ is the solution corresponding to a satisfying assignment, then there are $2^{N^2}$ solutions $r$ such that $d(r_1,r)=O(m)$, and at most $2^N$ further solutions where $d(r_1,r)$ can be as high as $m^{1/\epsilon'}+O(m)$, resulting in a total distance of at most $2^{N^2} \cdot O(m) + 2^N \cdot (m^{1/\epsilon'}+O(m))$. Note that the first of the two terms is in a much larger magnitude, at least if we assume that $N$ is only polynomially smaller than $m$, i.e. $N \geq m^{\delta}$ for some constant $\delta$ (otherwise, we can modify our generating group to contain $m^2$ branching gadgets instead). This means that we can upper estimate this expression by $2 \cdot 2^{N^2} \cdot O(m)$ for $m$ large enough.

On the other hand, if our algorithm finds a solution $r_2$ that does not satisfy $\phi$, then this has a total distance of at least $2^{N^2} \cdot m^{1/\epsilon'}$ due to the control group. This implies that the difference of centrality value between the two solutions is at least a factor of $\Theta(1) \cdot m^{1/\epsilon'-1}$, or in terms of $n$, at least $\Theta(1) \cdot n^{1-\epsilon'}$.

This is again asymptotically larger than $n^{1-\epsilon}$ for $n$ large enough, and hence any approximation algorithm must find a satisfying assignment for the formula.

Note that if there are more than $1$ satisfying assignments for $\phi$, then we can use the same argument, the difference between the two solutions only grows even larger.
\end{proof}

Note that there would be many other natural ways to express the fact that it is computationally hard to understand even the general distribution of the solution space: for example, we could say that given two specific solutions $r$ and $r'$, it is even NP-hard to decide whether $cent(r)>cent(r')$.

\subsection{Strictly better solutions}

Recall that for Theorem \ref{th:better}, we consider financial systems with loss, i.e. $\alpha \neq 1$ or $\beta \neq 1$.

\paragraph*{Unhappy penalty gadget.}
A main ingredient for the proof is the unhappy penalty gadget shown in Figure \ref{fig:unhappy}. Assume that there is a node $v$ in the system, and we want to add an outgoing penalty of some large weight $h$ to $v$, conditioned on the default of an indicator binary node $v_0$. If this task was executed by simply adding a CDS from $v$ to the sink $t$, then the solutions where $v_0$ is in default would not be strictly worse for every node in this subsystem, since $t$ would obtain a higher equity with the received penalty payment. In contrast to this, the unhappy penalty gadget ensures that the default of $v_0$ does provide a smaller-or-equal equity for each of the nodes.

Consider any parameters $\alpha, \beta<1$, and in terms of $\alpha$ and $h$, let us define a new parameter $b=\frac{h+5}{1-\alpha}$. The design of the unhappy penalty gadget requires us to add a CDS of weight $h$ towards a designated `semi-sink' node $t_0$, which has funds of $1$. However, we also add two further nodes $u$ and $t_0'$ to the gadget. Node $u$ has $b+1$ funds, a simple debt of weight $b$ to $t_0$, and an outgoing CDS of weight $2$ to $t_0'$, also in reference to $v_0$. Finally, $t_0$ has a simple debt of $b$ towards the sink node $t_0'$ (for a simpler analysis, we assume that $t_0'$ is not a general common sink in the system, but a distinct sink node specifically created for $v$; this does not affect our hardness result).

In this subsystem, if $v_0$ is not in default, then $u$ has no liability towards $t_0'$, and thus it is not in default; it pays its debt to $t_0$ and has an equity of $1$. Receiving this amount allows $t_0$ to pay its debt, thus also having an equity of $1$. As a result, the sink $t_0'$ receives a sum of $b$ in incoming payments from $t_0$.

On the other hand, if $r_{v_0}=0$, then $u$ has a total of $b+2$ liabilities, pushing it into default; thus, it can only use $\alpha \cdot (b+1)$ from its funds, paying $\frac{b}{b+2} \cdot \alpha \cdot (b+1) < \alpha \cdot (b+1)$ to node $t_0$ and $\frac{2}{b+2} \cdot \alpha \cdot (b+1) < 2$ to node $t_0'$. Even together with the sum of (at most) $h$ received from node $v$ as a penalty, this only gives total assets of less than $h+1+\alpha \cdot (b+1)$ for $t_0$. Note that our choice of $b$ ensures that $h+1+\alpha \cdot (b+1) < b$: since $b > \frac{h+2}{1-\alpha}$, we have
\[(1-\alpha) \cdot b > h+2 > h + 1 + \alpha, \]
implying
\[  b > h + 1 +\alpha + \alpha \cdot b = h+1+\alpha \cdot (b+1). \]
Thus node $t_0$ also cannot pay its liabilities in this case, and hence it is sent into default. This means that the sink $t_0'$ receives a payment of strictly less than $h+1+\alpha \cdot (b+1)$ from $t_0$, and together with the payment of strictly less then $2$ received from $u$, it has a total assets (and equity) of strictly less than $b$. This is again ensured by our choice of $b$: the fact that $b > \frac{h+4}{1-\alpha}$ implies 
\[ (1-\alpha) \cdot b > h+4 > h + 3 + \alpha,\]
which means that
\[  b > h + 3 +\alpha + \alpha \cdot b = h+1+\alpha \cdot (b+1) + 2. \]
Hence, $t_0'$ receives a total payment of strictly less than $b$ from the subsystem in this case.

Therefore, this latter solution is strictly worse for the whole subsystem: $v$'s equity is decreased due to the extra penalty of weight $h$, the nodes $v_0$, $u$ and $t_0$ have an equity of 0 now, and the equity of $t_0'$ is also smaller due to the smaller amount of incoming payments. Note that such a situation is only possible in systems with loss.

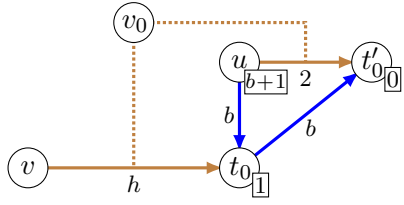
\begin{figure}
  \centering
	\input{unhappy.tikz}
	\caption{The unhappy penalty gadget}
	\label{fig:unhappy}
\end{figure}

\paragraph*{Proof of Theorem \ref{th:better}.}

Given the unhappy penalty gadget, we now describe the remaining details of the construction outlined in Section \ref{sec:better}. To avoid discussing infinite equities, we introduce a separate source node $s_0$ into this construction with $e_{s_0}=2$ only, and provide the incoming CDSs for node $w$ from this node.

As shown in Figure \ref{fig:subopt} of Section \ref{sec:better}, our construction is based on a pair of nodes $v_0$ and $v_0'$ with no funds and a debt of 1 to each other. Clearly the vectors $(0,0)$ and $(1,1)$ are solutions to this subsystem. In any other solution, both $v_0$ and $v_0'$ have to be in default, and thus any such other solution must satisfy $r_{v_0}=\beta \cdot r_{v_0'}$ and $r_{v_0'}=\beta \cdot r_{v_0}$. For any $\beta<1$ parameter, this only yields the solution $(0,0)$ again. Thus $v_0$ is indeed a binary node, and the subsystem acts as a different kind of branching gadget for the case of systems with loss.

Our construction for the theorem uses the base SAT construction, and then through unhappy penalty gadgets, it adds an arbitrarily large penalty to each node in the base construction with reference to $v_0$. Furthermore, the indicator nodes $v_{I}$ and $\overline{v_{I}}$ also receive such an unhappy penalty gadget with reference to $v_0$.

Finally, we have $w$ and $s_0$ in the construction, with $e_w=0$ and $e_{s_0}=2$. We add two distinct CDSs of weight $1$ from $s_0$ to $w$, one of them in reference to $v_0$, the other in reference to a node which indicates that $\phi$ is satisfied. Note that the illustration in Figure \ref{fig:subopt} is only a simplified sketch from this perspective; we cannot add this other CDS directly in reference to $\overline{v_I}$, since with $r_{v_0}=0$, the unhappy penalty gadgets do not ensure that $\overline{v_I}$ is a binary node, so this might provide some assets to $w$ even if $r_{v_0}=0$.

Instead, as a technical modification, we add an auxiliary node $z$ with $e_z=0$ and an incoming CDS in reference to the negation of $v_0$, and an outgoing CDS from $z$ to $w$ in reference to $\overline{v_I}$. This path of contracts provides no assets to $w$ if $r_{v_0}=0$. Note $q_z=0$ is ensured when $r_{v_0}=0$, and we can also ensure that the \textsc{not} gate attached to $v_0$ has no positive equity nodes for $r_{v_0}=0$ with an unhappy penalty gadget on its sink node.

In our reduction, the parameter solution $r$ is the one where $r_{v_0} = r_{v_0'}=0$, thus each node of the base construction is in default (with an equity of 0), and the nodes in the unhappy penalty gadgets are also not in a favorable state. Node $s_0$ has an equity of 1 in this solution. More importantly, node $w$ also has an equity of $1$, and thus any solution that is strictly better than $r$ must also have $q_w \geq 1$. Note that if $r_{v_0} = r_{v_0'}=0$ is fixed, then this is the only solution of the system.

Thus in any other solution, we must have $r_{v_0} = r_{v_0'}=1$. However, this implies that $w$ does not receive any payment through the CDS in reference to $v_0$. Hence a strictly better solution can only exist if it has $r_{\overline{v_{I}}}=0$ and thus $r_{v_I}=1$, i.e. if we find a satisfying assignment. Any such assignment indeed provides a strictly better solution: the nodes in the base construction cannot have less equity than 0, and the nodes in the unhappy penalty gadgets have strictly larger equities. Nodes $v_0$, $v_0'$ still have an equity of 0, and nodes $s_0$ and $w$ still have an equity of 1. Thus a strictly better solution than $r$ exists if and only if $\phi$ is satisfiable.

\section{Different optima for different objectives} \label{App:DiffOpt}

\subsection{Proof of Theorem \ref{obs:twoOpt}}

Let us first describe simple example systems that fulfill the properties outlined in Theorem \ref{obs:twoOpt}. For all pairs of objective functions $f_1$ and $f_2$, we apply a similar approach: we create a branching gadget to form two different solutions in the system, and we ensure that the optimum of $f_1$ is obtained when $r_x=1$, but on the other hand, choosing $r_y=1$ provides a much better solution in terms of $f_2$.

Let us now consider all the possible combinations of $f_1$ and $f_2$: 
\begin{itemize}
\item $f_1=$Min\textsc{Default}: for this case, we can simply use a bank $w$ with $e_w=0$ and an outgoing CDS in reference to $x$; this already ensures that $r_x=0$ results in a higher number of defaults than $r_x=1$.
\begin{itemize}
\item For $f_2=$Max\textsc{Prefer}, we add $\Theta(n)$ further nodes $u_i$ that have $e_{u_i}=2$, and an outgoing CDS of weight $1$ in reference to $y$. These new banks can never go into default (so they do not influence the optimum of Min\textsc{Default}), but if $r_y=0$, then their equities decrease from $2$ to $1$; as such $r_x=1$ gives a Max\textsc{Prefer} value of $\Theta(1)$, while $r_y=1$ gives a Max\textsc{Prefer} value of $\Theta(n)$.
\item For $f_2=$Min\textsc{Unpaid}, we add a single new node $u$ with $e_u=0$. This node will have an outgoing debt of $1$, and $\Theta(n)$ distinct outgoing CDSs of weight $1$, all in reference to $y$ (for this, we need to add $\Theta(n)$ distinct sinks to the system). Note that this node $u$ is in default in any case, so the optimum for Min\textsc{Default} is still obtained when $r_x=1$. However, now $r_y=1$ results in an unpaid debt of $O(1)$ altogether, while $r_y=0$ creates a total unpaid debt of $\Theta(n)$ in the system.
\end{itemize}
\item $f_1=$Max\textsc{Prefer}: let us choose a parameter $n'=\Theta(n)$, and add a large set of $n'$ nodes $w_i$ that all have $e_{w_i}=2$, and an outgoing CDS of weight $1$ in reference to $x$ (all going to the same sink $t$). This ensures that $r_x=1$ is the most preferred solution of at least $n'$ nodes.
\begin{itemize}
\item For $f_2=$Min\textsc{Default}, let us select a large constant $k$, and add $n'-k$ distinct nodes $u_i$ that have $e_{u_i}=1$ and an outgoing CDS of weight $2$ to the sink, in reference to $y$. Note that the system now consists of $n' + (n'\!-\!k) + O(1)$ nodes. If $r_x=1$, then $n'+O(1)$ banks are in default, but this is the primary preference of at least $n'$ nodes. On the other hand, if $r_y=1$, then only $O(1)$ banks are in default, but this solution is only preferred by $(n'\!-\!k)+O(1)$ banks. For a choice of a large enough constant $k$, this satisfies our requirements.
\item For $f_2=$Min\textsc{Unpaid}, it suffices to add a single bank $u$ with $e_u=0$, and $n'-k$ distinct outgoing CDSs of weight $1$ in reference to $y$, going to $n'-k$ distinct sink nodes (again for some large constant $k$). With $r_x=1$, the unpaid debt is $n'\!-\!k = \Theta(n)$, but this is the primary preference of at least $n'$ nodes. With $r_y=1$, the unpaid debt is only $O(1)$, but this solution is preferred by at most $(n'\!-\!k) + O(1)$ nodes.
\end{itemize}
\item $f_1=$Min\textsc{Unpaid}: we now use a bank $w$ with $e_w=0$ and some outgoing CDSs of weight $1$ in reference to $x$; however, the concrete number of these CDSs will now depend on our choice of $f_2$.
\begin{itemize}
\item For $f_2=$Min\textsc{Default}, we select a parameter $n'=\Theta(n)$, and add $n'$ outgoing unit-weight CDSs from $w$ (in reference to $x$). We then create $n'-k$ further nodes $u_i$ with $e_{u_i}=0$ and an outgoing CDS of weight $1$ in reference to $y$ (for some constant $k$). If $r_x=1$, this results in an unpaid debt of only $n'-k$, but yields $n'\!-\!k = \Theta(n)$ defaulting nodes. On the other hand, $r_y=1$ gives an unpaid debt of $n'$, but only results in $O(1)$ defaulting nodes.
\item For $f_2=$Max\textsc{Prefer}, we only add $k$ outgoing CDSs from $w$ (for some constant $k$), going towards $k$ distinct sink nodes. Besides this, we create $\Theta(n)$ banks $u_i$ that have $e_{u_i}=2$ and an outgoing CDS of weight $1$ in reference to $y$. With $r_x=1$, we now have $O(1)$ unpaid debts, but this is only the primary preference of $k+O(1)=O(1)$ nodes. With $r_y=1$, we have $k+O(1)$ unpaid debts, but this solution is preferred by $\Theta(n)$ nodes. This satisfies our requirements for a large enough constant $k$.
\end{itemize}
\end{itemize}

\subsection{A combined example}

We also show that we can merge these examples into a single construction that satisfies the properties outlined in Theorem \ref{obs:diffOpt}; this shows that it is even possible that all the optima are very far from each other simultaneously.

Furthermore, note that due to the CDSs in the network, the total amount of liabilities in the system may be drastically different in different solutions. Due to this, we also consider an alternative version of the Min\textsc{Unpaid} objective (termed Min\textsc{PropUnpaid}) in this example, where we minimize the proportion of unpaid liabilities compared to the total liabilities present in the system; i.e. we minimize $ ( \sum_{u,v \in B} l_{u,v} - p_{u,v} ) / ( \sum_{v \in B} l_v )$.

We also note that in contrast to our other results, this construction requires very large edge weights to create gaps between all pairs of functions, and allow a straightforward analysis at the same time. As such, this example does not generalize to the bounded edge-weight case in a trivial way. 

We prove Theorem \ref{obs:diffOpt} in the following form:

\begin{theorem*}
Let $h$ be an arbitrarily large number; for convenience, we assume $h=\omega(n)$. There exists a financial system with exactly four solutions $r_1$, $r_2$, $r_3$ and $r_4$, such that:
\begin{itemize}[topsep=4pt, itemsep=4pt, parsep=0pt]
 \item in terms of Min\textsc{Default}, $r_1$ is an $\Omega(\sqrt{n})$ factor better than $r_2$, $r_3$ and $r_4$,
 \item in terms of Max\textsc{Prefer}, $r_2$ is an $\Omega(\sqrt{n})$ factor better than $r_1$, $r_3$ and $r_4$,
 \item in terms of Min\textsc{Unpaid}, $r_3$ is an $\Omega(h)$ factor better than $r_1$, $r_2$ and $r_4$,
 \item in terms of Min\textsc{PropUnpaid}, $r_4$ is an $\Omega(h)$ factor better than $r_1$, $r_2$ and $r_3$.
\end{itemize}
\end{theorem*}

\renewcommand{\proofname}{Proof.}

\begin{proof}
The different parts of our proof construction are illustrated in Figure \ref{fig:diffOpt}. Creating a system that has exactly 4 solutions is straightforward: we use 2 branching gadgets that together provide 4 combinations of states. We can then use \textsc{and} gates to create four indicator binary nodes $u_1$, $u_2$, $u_3$, $u_4$ for each of these combinations. In each solution of the system, exactly one of the four nodes $u_i$ has $r_{u_i}=1$.

\begin{figure}[t]
\centering
	\resizebox{0.99\textwidth}{!}{\input{diffOpt.tikz}}
	\caption{Example system where the optima for different objective functions are realized in different solutions. For simplicity, the different parts of the construction are illustrated separately.}
	\label{fig:diffOpt}
\end{figure}
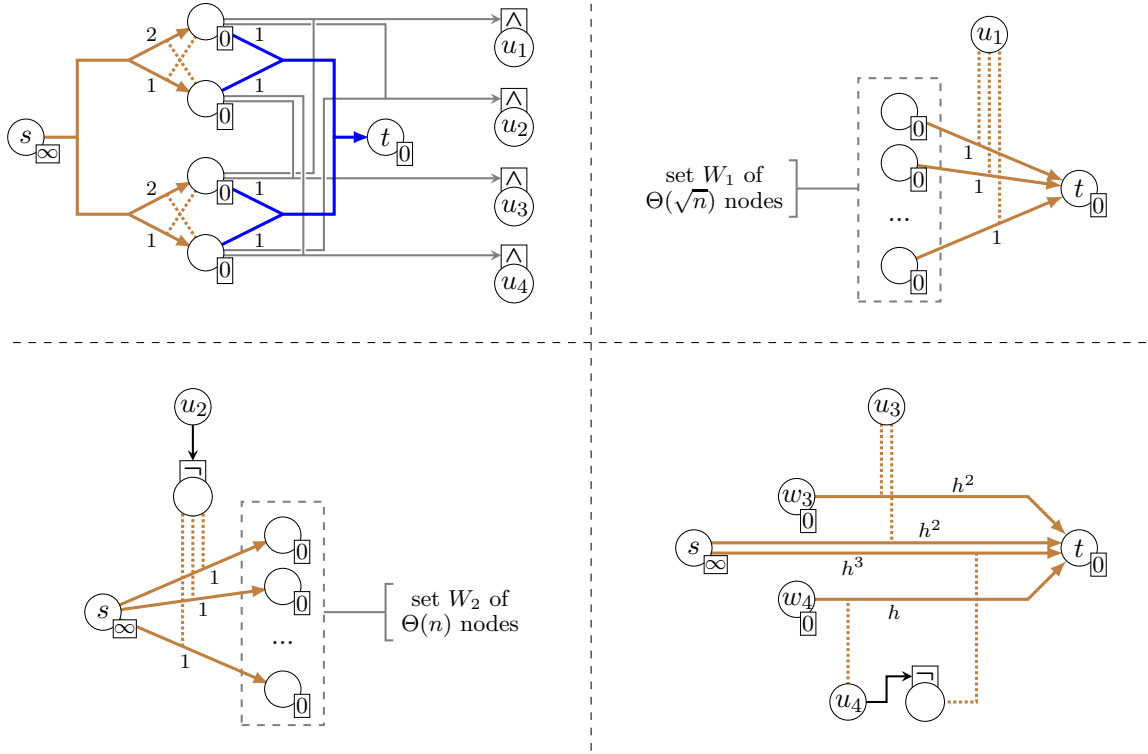

We then attach four different sets of nodes to the four indicator nodes in order to ensure that each solution has the desired properties.

For the case of $u_1$, we add a set $W_1$ of $\Theta(\sqrt{n})$ distinct nodes to the system, which all have 0 funds. From each of these nodes, we create a CDS to $t$ with a weight of 1, in reference to $u_1$. Thus if $u_1$ is chosen, then none of these $\Theta(\sqrt{n})$ nodes have any liabilities, and they all survive. On the other hand, if we choose any other solution, then these CDSs all incur liabilities, and hence our system has $\Theta(\sqrt{n})$ nodes in default. Besides $W_1$, the system will only have $O(1)$ nodes that can ever go into default, so in the solution where $r_{u_1}=1$, the number of defaulting nodes is only $O(1)$. Thus $u_1$ is indeed a factor of $\Omega(\sqrt{n})$ better than any other solution in terms of Min\textsc{Default}.

To ensure that $u_2$ is the first preference of $\Theta(n)$ nodes, we add a set $W_2$ of $\Theta(n)$ new nodes to the system, all with 0 funds. We then create a CDS from $s$ to each of these nodes with a weight of 1, in reference to the negation of $u_2$. If $u_2$ is chosen, then these $\Theta(n)$ nodes all have an equity of $1$; otherwise, their equity is 0. Since the rest of our system will only contain $O(\sqrt{n})$ nodes, this shows that selecting $u_1$ is the primary preference of the $\Theta(n)$ nodes in $W_2$, while all other solutions are the primary preference of at most $O(\sqrt{n})$ nodes, and thus $u_2$ is indeed an $\Omega(\sqrt{n})$ factor better in terms of Max\textsc{Prefer}.

Finally, we analyze the case of $u_3$ and $u_4$ together. For $u_3$, we only create one new node $w_3$ with $e_{w_3}=0$, and add two new CDSs to the system. Both of these CDSs are in reference to $u_3$, going to the sink $t$, and have a weight of $h^2$; one of them comes from $w_3$, the other comes from $s$. For $u_4$, we add a single node $w_4$ with $e_{w_4}=0$ again, and we create two new CDSs going to $t$. The first CDS comes from $w_4$, has a weight of $h$, and is in reference to $u_4$. The second CDS comes from $s$, has a weight of $h^3$, and is in reference to the negated version of $u_4$.

This means that if any solution other than $u_3$ is chosen, then we introduce $h^2$ paid and $h^2$ unpaid liabilities into the system. Similarly, if $u_4$ is chosen, then $h^3$ paid liabilities are introduced, but if $u_4$ is not chosen, then $h$ unpaid liabilities are introduced. In contrast to this, the CDSs based on $u_1$ and $u_2$ only result in $O(n)$ paid or unpaid liabilities, so since we assume $h=\omega(n)$, the total amount of liabilities is always determined by the CDSs of $u_3$ and $u_4$.

Let us analyze the total amount of paid and unpaid liabilities in all four solutions. If $u_1$ or $u_2$ is chosen, then the CDSs of $u_3$ ensure that there is a $\Theta(h^2)$ amount of both paid and unpaid liabilities in the system. If $u_3$ is chosen, then the amount of unpaid liabilities is only $\Theta(h)$, while the amount of paid liabilities is $O(n)$. Finally, if $u_4$ is chosen, the amount of unpaid liabilities is $\Theta(h^2)$, while the total amount of paid liabilities is $\Theta(h^3)$.

This shows that $u_3$ and $u_4$ indeed fulfill our requirements. The total amount of unpaid liabilities is $\Theta(h)$ in $u_3$, and $\Theta(h^2)$ in all other solutions, which is indeed a difference of a factor $\Omega(h)$ in terms of Min\textsc{Unpaid}. The rate of unpaid to total liabilities is a constant in $u_1$ and $u_2$, and asymptotically $1$ in case of $u_3$, but it is only $\Theta(\frac{1}{h})$ in case of $u_4$. Thus, $u_4$ is indeed a factor of $\Omega(h)$ better than all other solutions in terms of Min\textsc{PropUnpaid}.
\end{proof}

\section{Different model variants} \label{App:B}

\subsection{Systems with loss} \label{App:loss}

While our constructions were presented for $\alpha=\beta=1$, our results also hold for any $\alpha, \beta \in (0,1]$. Moreover, a different choice of $\alpha, \beta$ only requires minor modifications to our hardness proofs.

Firstly, note that the choice of $\alpha$ and $\beta$ only affects the behavior of nodes $v$ that have $0<r_v<1$. Since the majority of our gadgets work with binary nodes, they require no modification for any choice of $\alpha, \beta$. Recall that \textsc{and} gates can be replaced by combinations of \textsc{not} and \textsc{or} gates, and cutoff gadgets are not used in our constructions. The only building block we need to modify is the clean branching gadget: one can observe that a choice of $\delta_x=\frac{1}{\beta}$ and $\delta_y=1$ provides a similar tool of binary choice for any $\alpha, \beta$.

The nodes representing the objective functions in our constructions are also binary nodes, so they require no changes either. The only exception to this is the unhappy penalty gadget in Theorem \ref{th:better}, but this was already defined with respect to a specific $\alpha, \beta$ in a default cost setting.

\subsection{Unit-weight contracts} \label{App:bounds}

We now discuss how to adapt our results to the case when we are only allowed to use debts and CDSs of weight $1$. We have already noted that in most cases when the weight of a contract is not $1$ (but a larger integer $k$), we can usually split this into $k$ distinct contracts that come from/go to $k$ distinct source/sink nodes. In particular, we can apply this on the incoming CDS of $x$ in the clean branching gadget, or on node $v$ of the Max\textsc{Equity} reduction.

Note that if we also want to reduce $e_v$ in the Min\textsc{Equity} case to a constant, we can similarly do this by introducing $\Theta(n)$ new source nodes that are debtors of $v$.

Removing the infinitely large funds and weights in the construction of Theorem \ref{th:repr} is also straightforward, as we have already noted in Appendix \ref{App:repr}.

The only more involved case is Theorem \ref{th:better}, and in particular, the unhappy penalty gadgets. Note that in this gadget, for any integer weight $k$, we can again replace a contract of weight $k$ by $k$ new intermediate nodes: e.g. for the contract from $u$ to $t_0$, we create $b$ intermediate nodes with no funds that have a unit-weight incoming debt from $u$ and a unit-weight outgoing debt towards $t_0$. Thus in order to adapt the gadget to this setting, we only have to ensure that the parameters $h$ and $b$ are integers.

Recall that the choice of $h$ is entirely up to us: we just need to select it large enough such that it sends the corresponding node of the base construction into default, which can always be done with a constant integer value. On the other hand, the value of $b$ only needs to satisfy $b>\frac{h+4}{1-\alpha}$, so any integer value above this threshold suffices.

\subsection{The dependency graph} \label{App:depend}

In order to gain a deeper understanding of the relations between nodes, the authors of \cite{base1, base2} introduce the so-called \emph{colored dependency graph}. This graph is essentially a transformation of the financial system which removes the ternary relations (i.e. CDSs), and instead models the system as a simple directed graph with edges of two colors: red and green. Intuitively, a green edge from $u$ to $v$ means a long position, i.e. that it is better for $v$ if $u$ has a larger recovery rate. On the other hand, a red edge from $u$ to $v$ means a short position, i.e. that a smaller recovery rate at $u$ is more beneficial to $v$. We now outline the formal definition for the colored dependency graph; see \cite{base1} for more details.

\begin{figure}
\centering
	\input{dependency.tikz}
	\caption{Illustration of the transformation into the dependency graph, in line with \cite{base1}. Green arrows (with filled arrowheads) express long positions, while red arrows (with empty arrowheads) express short positions.}
	\label{fig:depend}
\end{figure}
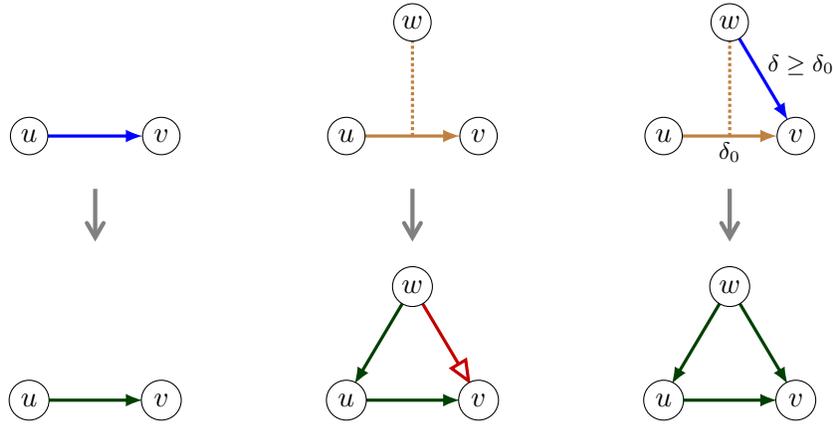

The dependency graph has the same node set as the original financial system. The edges of the dependency graph are then formed according to the following rules:
\begin{itemize}[topsep=3.5pt, itemsep=3.5pt, parsep=0pt]
 \item If there is a debt or CDS from $u$ to $v$, we draw a green edge from $u$ to $v$.
 \item If there is a CDS from $u$ in reference to $w$, we draw a green edge from $w$ to $u$.
 \item If the incoming CDSs of $v$ in reference to $w$ have a total weight of $\sum_{u \in B} \, l_{u,v}^w=\delta_0$ in the system, and we have $l_{w,v}<\delta_0$, we draw a red edge from $w$ to $v$.
\end{itemize}
An illustration of these rules, originally from \cite{base1}, is visible in Figure \ref{fig:depend}. The rule set provides a simple directed graph with two-colored edges. Note that given CDSs $l_{u_i,v}^w$ from nodes $u_i$, a red edge from $w$ to $v$ is only added if the CDSs together have larger weight than the debt from $w$ to $v$; otherwise, a higher $r_w$ value is more beneficial for $v$ altogether, so we call these CDSs \textit{covered}.

Based on the dependency graph, the work of \cite{base1} discusses $3$ restricted classes of financial systems that always guarantee the existence of a solution. The most general of these classes is the class that we refer to as \textit{RFC systems}, when no directed cycle in the dependency graph contains a red edge. This model already ensures the existence of a solution \cite{base1}, while still allowing reasonably rich behavior, so it might be of particular interest.

For the general intuition behind this restricted model, we can consider the strongly connected components (SCCs) of the dependency graph. Within each such SCC we only have green edges, which ensures that if the recovery rate outside this component is fixed, then there exists a maximal solution within this component. Given a topological ordering of SCCs, every component is only affected by the preceding components in the ordering. Therefore, we can iterate through the components in this order, and always select the maximal solution in the current component, considering the recovery rates of the preceding components to be already fixed.

\subsection{Hardness results for RFC systems} \label{App:rfc}

We now discuss the proof of Theorem \ref{th:rfc}, i.e. adapting our constructions to the restricted case of RFC systems.

We have already noted that the key observation is that, in fact, directed cycles are very rare in the dependency graphs of our constructions: we mostly use logical gates that follow a specific ordering, and thus the dependency graphs are already very close to DAGs. The only exceptions are branching gadgets, where $x$ and $y$ both have a CDS in reference to each other, and hence there is a red edge between them in both directions. Indeed, $x$ and $y$ clearly have a short position on each other.

Hence we only need to devise a branching gadget that has the same functionality, while also satisfying the RFC property. This modified branching gadget is illustrated in Figure \ref{fig:modBranch}. The gadget is based on two nodes that are connected as shown in Figure \ref{fig:cyc} earlier; we now term them $v_0$ and $v_0'$. These nodes have no funds, a debt of 1 to each other, and are not affected by any other banks. The solutions of this subsystem are exactly the clearing vectors with $r_{v_0}=r_{v_0'}=\rho$ for some $\rho \in [0,1]$.

The key idea behind our approach is to ensure that in any reasonable solution, this subsystem obtains either a very small or a very large $\rho$ value; then the value of $\rho$ being small or large indicates the binary choice that was previously represented by our branching gadgets.

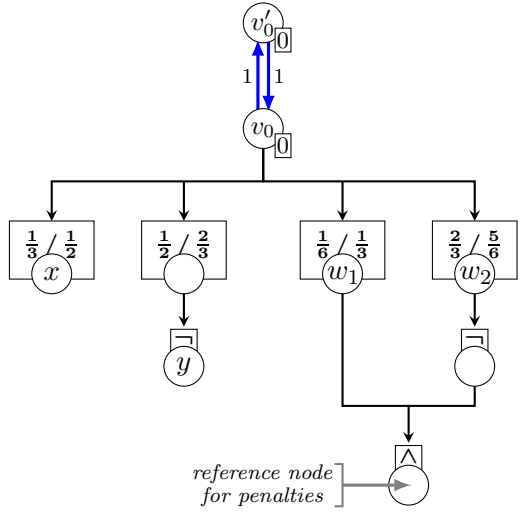
\begin{figure}[t]
\centering
	\input{cycleBranch.tikz}
	\caption{Modified branching gadget for RFC systems}
	\label{fig:modBranch}
\end{figure}

We ensure that in every reasonable solution, we have either $\rho \leq \frac{1}{3}$ or $\rho \geq \frac{2}{3}$ in this branching gadget. To achieve this, we append two cutoff gadgets to $v_0$, one with parameters $\eta_1=\frac{1}{6}$ and $\eta_2=\frac{1}{3}$, and the other with $\eta_1=\frac{2}{3}$ and $\eta_2=\frac{5}{6}$. Let the output nodes of these cutoff gadgets be denoted by $w_1$ and $w_2$, respectively. We then encode the condition $w_1$ \textsc{and} (\textsc{not} $w_2$) through logical gates, and if this condition holds, we introduce a very large penalty to our objective function (discussed in detail below).

This already ensures that we cannot have $\rho \in [\frac{1}{3}, \frac{2}{3}]$ in any reasonable solution. If an approximation algorithm would return such a solution, then we could easily improve the objective function value by simply selecting an arbitrary assignment of variables (that either satisfies $\phi$ or it does not), and consider this assignment instead, thus obtaining a strictly better approximation algorithm.

Finally, we add two further cutoff gadgets to node $v_0$, with parameters $\eta_1=\frac{1}{3}$, $\eta_2=\frac{1}{2}$ and parameters $\eta_1=\frac{1}{2}$, $\eta_2=\frac{2}{3}$. If we denote the output of the first gadget by $x$, and the negation of the output of the second gadget by $y$, then we can use $x$ and $y$ to represent a binary choice like our original branching gadget. These gadgets ensure that if $\rho \leq \frac{1}{3}$, then $r_{x}=1$ and $r_{y}=0$, and if $\rho \geq \frac{2}{3}$, then $r_{x}=0$ and $r_{y}=1$.

We note that when $\rho \in (\frac{1}{6}, \frac{1}{3})$ or $\rho \in (\frac{2}{3}, \frac{5}{6})$, then the cutoff gadgets do not ensure that $w_1$ and $w_2$ are binary nodes, and thus our penalty indicator node might also not be binary. This means that for this alternative branching gadget, some fraction of the penalty might already apply when $\rho \in (\frac{1}{6}, \frac{1}{3})$ or $\rho \in (\frac{2}{3}, \frac{5}{6})$. However, any such solution can be improved by selecting $\rho \leq \frac{1}{6}$ or $\rho \geq \frac{5}{6}$ instead, respectively; thus, we could even assume that reasonable solutions always have $\rho \notin (\frac{1}{6}, \frac{5}{6})$. Note that we are indirectly using two properties in this claim: (i) that the recovery rate of the cutoff gadget's output is monotonic in the input even on the interval $[\eta_1, \eta_2]$, and (ii) that the penalties we introduce for the specific objectives are also monotonic in the recovery rate of our penalty indicator node.

Also note that our alternative branching gadget becomes significantly simpler for the case of $\beta<1$. In this case, the only valid solutions to the subsystem consisting of $v_0$ and $v_0'$ are $(0,0)$ and $(1,1)$, so $v_0$ itself can already represent the binary choice without the cutoff gadgets.

\paragraph*{Penalty functions for the objectives}

For most of the objective functions we have studied, it is rather straightforward to add a very large penalty in case our restrictions on any of the alternative branching gadgets is violated. For convenience, we assume that the penalty indicators of the branching gadgets are first merged into a single node with an \textsc{or} gate, and then negate it; this node indicates whether all the branching gadgets are ``initialized'' properly.

For example, if we add $n$ new nodes with funds of $0$ and an outgoing CDS of weight $1$ (in reference to this penalty indicator), then this already suffices for the Min\textsc{Default} and Min\textsc{Unpaid} objectives. Any incorrect initialization of a branching gadget will then result in $n$ extra defaults and an extra unpaid debt of $n$, so any such solution is only improved if we replace it by an arbitrary variable assignment. If we change the funds of the nodes to $1$ and the CDS weight to $2$, then this also works for Min\textsc{LeastPrefer}. Note that given our original construction on $n$ nodes, we can indeed add $O(n)$ new nodes to this system without changing its basic properties: all functions values will still have the same magnitude compared to the network size.

For Min\textsc{Equity}, we can simply add a large-weight incoming CDS to $v$ if the penalty applies; the same approach also works with Min\textsc{Diff}.

For maximization problems, we use a different approach to avoid adding a value of $\Theta(n)$ to any solution. Here we simply add a large-weight outgoing CDS to every bank of our original system, in reference to the penalty node. This is sufficient for both Max\textsc{Equity} and Max\textsc{Prefer}, as well as for Max\textsc{Surviving}.

Max\textsc{Paid} requires a slightly different approach: we consider the extra nodes $u_i$ that implement the objective function, and instead of giving them funds of $1$, we give them an incoming CDS of weight $1$ in reference to the negation of the penalty indicator node. If the constraints are violated, then even if we find a satisfying assignment of $\phi$, this results in no paid liabilities in the network, since the banks $u_i$ have no assets at all. As such, if $\phi$ is satisfiable, then an approximation algorithm must return a solution where the branching gadgets are initialized properly.

In case of Theorem \ref{th:repr}, we add a further, even larger control group where the recovery rates are always $1$, unless the penalty indicator is set to true; thus any arbitrary assignment of $\phi$ is more optimal than violating a constraint.

In case of Theorem \ref{th:better}, we add an outgoing CDS from $w$ with weight $2$ in reference to the penalty node.

\subsection{Green systems and regularity}

As discussed before, the alternative branching gadget in Theorem \ref{th:rfc} uses a debt-only network with multiple solutions to create a large solution space. As such, we now take a detour to study such systems, or more generally, any systems with only long positions (i.e. greens systems). Recall that a green system can only contain simple debts and so-called covered CDSs (see the definition of the dependency graph). Theorem \ref{th:ambiguity} shows that in these systems, default ambiguity can only happen in very special cases.

Before the proof, we note that one might wonder if the second observation in Theorem \ref{th:ambiguity} can also be phrased with a path of \textit{green edges} instead of a path of contracts. However, this is not the case. Essentially, the property we require for the proof is that $a_{u}>0$ already implies $a_{v}>0$; this indeed holds if we have a debt or a covered CDS from $u$ to $v$. On the other hand, a green edge from $u$ to $v$ can also be present because $v$ has an outgoing CDS in reference to $u$ (see the second point in the definition of the dependency graph), which does not satisfy this property.

The early work of Noe has already studied default ambiguity in debt-only networks, showing that the solution is unique if from any bank there is a directed path to another bank with positive funds \cite{model1}. We now prove a more general version of this result, generalizing the theorem to any green system. In particular, we show that green systems can only have multiple solutions in a special edge case: when we essentially have a strongly connected component of positive liabilities, with no funds and no incoming assets at all.

The first step of the proof is to note that a green system ensures the existence of a maximal and minimal solution. This has already been proven in the work of \cite{base1}: intuitively, the payment functions are monotonous due to the long positions, so one can use the Knaster-Tarski fixed point theorem to prove that a maximal solution $r$ and minimal solution $r'$ exists. This maximal (minimal) solution assigns the highest (lowest) recovery rate to every bank simultaneously.

We first outline the rest of the proof for debt-only networks, and then we separately discuss the changes required for the case of green systems. The proof for debt-only networks is basically the same as in the analysis of Eisenberg and Noe \cite{model1}, extended with some further observations.

\renewcommand{\proofname}{Proof for debt-only networks.}

\begin{proof}
Using $r$ and $r'$, one can first show that any bank $v$ must have $q_v=q_v'$ in these two solutions $r$ and $r'$. Intuitively, since the sum of equity is fixed in any solution (see Lemma \ref{ops:pareto}), if a bank had lower equity in $r'$ than in $r$, then another bank would need to have a larger equity in $r'$ than in $r$, which is a contradiction. Since the equity of each bank $v$ is maximized in $r$ and minimized in $r'$, this also implies that $v$ has the same equity in every solution of the system.

Now let us assume that a bank $v_0$ has two different recovery rates in two solutions of the system; since $r$ and $r'$ contain the maximal and minimal recovery rates for $v_0$, this also implies $r_{v_0}>r_{v_0}'$. As such, it is enough to show that our observations hold if $r_{v_0}>r_{v_0}'$.

Since $v_0$ is still in default in the solution $r'$, there will be strictly more payment on the outgoing debts of $v_0$ in $r$. As the system is loss-free, each defaulting bank will relay these extra payments, and they will traverse the network until either reaching a bank $u$ with $r_u'=1$, or arriving back at $v_0$. However, the former is also not possible: if a bank with $r_u'=1$ receives more incoming payment, then its equity strictly increases, which contradicts the previous observations.

This means that every directed path of debts starting from $v_0$ must lead back to $v_0$, implying that the nodes reachable from $v_0$ form an SCC $C$. Moreover, note that it must also hold that each bank $u \in C$ has $q_u=0$ in both solutions. This also implies that $e_u=0$ for any $u \in C$: since no outgoing payment leaves the SCC, we would otherwise have a positive equity in at least one of the banks in $C$ (again due to the argument in Lemma \ref{ops:pareto}). Similarly, having $e_w>0$ at a bank $w$ with a path of debts to $C$ implies that a positive amount of assets arrive in $C$, so again at least one bank in $C$ would have a positive equity.
\end{proof}

Hence for debt-only networks, the nodes reachable from $v_0$ form an SCC $C$ with no funds and no incoming assets at all.

\renewcommand{\proofname}{Proof for green systems.}

\begin{proof}
Adapting the same argument to green systems is not straightforward, since we also have CDS contracts, and these might carry less payment in $r$ than in $r'$ if the reference entity has a higher recovery rate in $r$. As such, it is not immediately clear that a higher $r_u$ value at $u$ always translates to strictly more payment for the creditors of $u$.

This could be a problem for our proof: if there are funds in $C$ that indirectly provide an equity to a node $v$, but $v$ does not receive extra payments in $r$, then we do not have our previous contradiction. As such, the technical part of the proof is to prove the following claim in green systems.

\begin{lemma} \label{lem:tech}
If $r_u>r_u'$ and $p_{u,v}'>0$ in $r'$, then $a_v>a_v'$. 
\end{lemma}

Since we have $l_v \leq l_v'$ for any liability, $a_v>a_v'$ will then imply $r_v > r_v'$ if $v$ is still in default in $r'$.

We can then again use $C$ to denote the set of banks that are reachable from $v_0$ on a path of contracts that all have positive payment in $r'$. As before, an inductive argument shows that for any node $u \in C$ we must have $a_u>a_u'$.

This once again ensures that each bank $u \in C$ has $e_u=0$: otherwise, since there are no outgoing payments from $C$ in $r'$, we would have a bank $v$ in $C$ with positive equity $q_v'$. Then $a_v>a_v'$ would imply $q_v>q_v'$, which is a contradiction. Similarly, having $e_w>0$ at a bank $w \notin C$ but with a path of contracts to $C$ would imply that a positive amount of assets arrive in $C$, so again at least one bank in $C$ would have a positive equity.
\end{proof}

Recall that for debt-only networks, we have also noted that the nodes reachable from $v_0$ form an SCC. Making an analogous statement is not so straightforward in this case, since the same claim only holds for the contracts that have a positive liability; e.g. if there is a CDS from $v_0$ to a bank $u$ such that the liability on the CDS is $0$ in any solution, then no restrictions follow for $u$. As such, in this case, it is not straightforward to find the component $C$ that is affected by the extra payments; we have to compute $r'$ in order to do so. Due to this, we have limited the scope of Theorem \ref{th:ambiguity} to nodes reachable on debt contracts.

Altogether, the solution of a green system is known to be unique if we ensure that for every bank $v_0$, there is either (i) a bank $u$ reachable from $v_0$ on a path of debt contracts, with $e_u>0$, or (ii) a bank $w$ with $e_w>0$ such that $v_0$ is reachable from $w$ on a path of contracts of any kind.

Regularity is the simplest way to ensure this property, but we can also come up with weaker conditions instead: e.g. we could specifically say that a system is \emph{path-regular} if for all $v \in B$, there exists a bank $w$ with $e_w>0$ such that there is a path of contracts from $w$ to $v$ in $G$. However, enforcing this would be much more difficult from a regulator's perspective.

Finally, let us discuss the proof of our technical lemma.

\renewcommand{\proofname}{Proof of Lemma \ref{lem:tech}.}

\begin{proof}
We first create an auxiliary network that separates the effects of each covered CDS contract. Given a CDS $\delta_{u,v}^w$, we introduce a fictitious node $z$ for this specific CDS: we (i) set $z$ to be the new creditor of the CDS $\delta_{u,v}^w$, (ii) add a debt of weight $\delta_{u,v}^w$ from $w$ to $z$, (iii) decrease the debt from $w$ to $v$ by $\delta_{u,v}^w$, and (iv) introduce an infinite liability from $z$ to $v$. This bank $z$ essentially captures the payments we attribute altogether to the CDS and the ``part'' of the debt from $w$ to $v$ that covers it. The payment is visible as $p_{z,v}$. Note that each CDS is covered in our green system, so we can introduce such an auxiliary node for all incoming CDSs of $v$ in reference to $w$, and the weight of the debt from $w$ to $v$ will still remain non-negative.

Now let us begin the proof of $a_v>a_v'$. First, note that since payments on each debt contract are monotonic in the recovery rates, each incoming debt of $v$ has at least as much payment in $r$ as in $r'$. Also, for a covered CDS $\delta_{u,v}^w$, if $r_u \geq r_u'$ and $r_w \geq r_w'$, then we have $p_{z,v} \geq p_{z,v}'$ from the corresponding auxiliary node $z$. This already implies $a_v \geq a_v'$. It only remains to show that if $r_u>r_u'$ and $p_{u,v}'>0$, then the contracts from $u$ to $v$ indeed carry strictly more payments, so we have $a_v>a_v'$.

If the payment from $u$ to $v$ happens (partially) on a simple debt contract, then this is straightforward: $r_u>r_u'$ implies a strictly higher payment on any such contract.

If the payment happens on a CDS in reference to $w$, and the payment becomes strictly larger in the solution $r$ (e.g. because $r_w=r_w'$), then we are again finished. This leaves the more involved case when the payment on this CDS does not increase: this can indeed happen if we have $r_w>r_w'$, and thus the liability on the CDS is smaller in $r$.

Hence assume that we have $r_u \cdot \delta_{u,v}^w \cdot (1-r_w) \leq r_u' \cdot \delta_{u,v}^w \cdot (1-r_w')$. First note that since $r_u>r_u'$, this implies $r_w>r_w'$. Let $z$ denote the auxiliary node for this CDS; we need to show that even though the payment on the CDS did not increase, we still have $a_z>a_z'$. For this, we need to show for the sum of payments that
\[ r_u \cdot \delta_{u,v}^w \cdot (1-r_w) + \delta_{u,v}^w \cdot r_w > r_u' \cdot \delta_{u,v}^w \cdot (1-r_w') + \delta_{u,v}^w \cdot r_w'. \]
The definition of $z$ allows us to conveniently remove the coefficients $\delta_{u,v}^w$ from each term. We expand the brackets to obtain
\[ r_u +  r_w - r_u \cdot r_w > r_u' +  r_w' - r_u' \cdot r_w'. \]
Adding $-1$ and reorganizing provides $(1-r_u) \cdot (1-r_w) < (1-r_u') \cdot (1-r_w')$, which indeed holds since $r_u>r_u'$ and $r_w>r_w'$.
\end{proof}

\subsection{Hardness results for regular systems}

We now discuss the proof of Theorem \ref{th:reg}, i.e. that our hardness results carry over to the case when each bank $u$ has $e_u>0$.

\renewcommand{\proofname}{Proof of Theorem \ref{th:reg}.}

\begin{proof}
We have already outlined that the main idea is to consider a new pair of binary states at a bank $v$: we still have $r_v=1$ as one of the two states, but now the other state will be represented by $r_v=\frac{1}{2}$. This will allow us to give some funds to every node in our construction, thus fulfilling the regularity condition.

In fact our main building blocks are surprisingly easy to adapt to this setting: one can observe that by setting $e_x=e_y=0.5$ in the clean branching gadget, setting $e_w=0.5$ in the \textsc{not} gate and also $e_w=0.5$ in the \textsc{or} gate, all of these gadgets will exhibit the same functionality as before (assuming, of course, that the inputs now also follow this new binary representation). Any \textsc{and} gates can be replaced by a combination of \textsc{not} and \textsc{or} gates, and the cutoff gadgets are not used in the base versions of our hardness proofs.

As such, these simple changes already allow us to adapt the base construction to the case of regular systems. It remains to discuss the modifications in the rest of the system for each of our hardness results.

For both of the equity objectives, we can simply set $e_v=1$ and add a new outgoing debt of weight $1$ from $v$; this does not affect the equity of $v$ in any solution. This also settles the reduction for Min\textsc{Diff}, where we execute this for both $v_1$ and $v_2$.

In the construction used for Min\textsc{Default}, Max\textsc{Surviving} and Min\textsc{Unpaid}, we can simply set the funds of the extra nodes to $e_{u_i}=0.5$; this still ensures that the attached CDS will result in defaults and unpaid debts. On the other hand, the extra nodes in the constructions for Max\textsc{Prefer}, Min\textsc{LeastPrefer} and Max\textsc{Paid} already satisfy regularity.

Adapting the system in Theorem \ref{th:repr} is a more difficult task. Fortunately, the nodes in the control group do not cause any problem in this setting: we can change their funds to some small value, e.g. $e_{u_i}=0.1$, so their recovery rates in the two cases will be $r_{u_i}=0.1$ and $r_{u_i}=1$. This still ensures that an unsatisfying assignment results in a total distance of at least $0.9 \cdot 2^{N^2} \cdot m^{1/\epsilon'}$, and thus we can apply the same argument as before.

On the other hand, the branching gadgets in the generating group require more attention. Note that we only want these gadgets to introduce a binary choice when their source nodes receives funds, but otherwise, we want them to have a single solution only. One can show that this happens if we provide e.g. $\frac{1}{4}$ funds to their source node: a clean branching gadget with $e_s=\frac{1}{4}$ and $e_x=e_y=\frac{1}{2}$ indeed has only one solution (at $r_x=\frac{1}{4}+\sqrt{3}/4$, $r_y=1-\sqrt{3}/4$). On the other hand, if the source node receives $2$ more assets, then it behaves like a regular branching gadget as discussed above.

Theorem \ref{th:better} is again a more involved case, since we have to adapt our gadgets to the case of systems with loss. We discuss this construction for $\alpha=\beta=0.5$. One can observe that we can again use the original clean branching gadget, \textsc{not} gate and \textsc{or} gate with minor modifications: we now set the funds of the corresponding nodes to $\frac{2}{3}$ instead of $\frac{1}{2}$. This way, when the nodes are in default, they only have a recovery rate of $\frac{1}{3}$ after applying default costs, so in this case, our binary states will be represented by recovery rates of $1$ and $\frac{1}{3}$.

Adapting the rest of the construction is a simpler task: we can set $e_{v_0}=e_{v_0'}=\frac{1}{3}$ to ensure that the subsystem has two solutions $r_{v_0}=r_{v_0'}=1$ and $r_{v_0}=r_{v_0'}=\frac{1}{3}$, as discussed at the end of Section \ref{sec:models}. We also set $e_w=1$. Note that in the unhappy penalty gadgets, every node (except for the sink) has funds already, so this requires no major modification; we only need to scale up the weights of the CDSs to account for the fact that the lower binary state is now represented by $r_{v_0}=\frac{1}{3}$ instead of $r_{v_0}=0$.
\end{proof}

\end{appendices}

\end{document}

%% file: example.tikz

\begin{tikzpicture}
	
	\draw[very thick, blue, arrows=-latex] (0pt,0pt) -- (72pt,0pt);
	\draw[very thick, blue, arrows=-latex] (0pt,0pt) -- (36pt,62.1pt);
	\draw[very thick, brown, arrows=-latex] (40pt,69pt) -- (76pt,6.9pt);
	\draw[very thick, brown, densely dotted] (60pt,34.5pt) -- (0pt,0pt);
	
	\node[anchor=center] at (40pt,-7pt) {\small $2$};
	\node[anchor=center] at (15pt,37pt) {\small $2$};
	\node[anchor=center] at (65pt,37pt) {\small $2$};
	
	\draw[black, fill=white] (0pt,0pt) circle (1.6ex);
	\draw[black, fill=white] (80pt,0pt) circle (1.6ex);
	\draw[black, fill=white] (40pt,69pt) circle (1.6ex);
	
	\node[anchor=center] at (0pt,0pt) {\normalsize $u$};
	\node[anchor=center] at (80pt,0pt) {\normalsize$v$};
	\node[anchor=center] at (40pt,69pt) {\normalsize $w$};
	
	\draw [fill=white] (4pt,-2pt) rectangle (10pt,-11pt);
	\node[anchor=center] at (7pt,-6.5pt) {\footnotesize $2$};
	\draw [fill=white] (84pt,-2pt) rectangle (90pt,-11pt);
	\node[anchor=center] at (87pt,-6.5pt) {\footnotesize $1$};
	\draw [fill=white] (44pt,67pt) rectangle (50pt,58pt);
	\node[anchor=center] at (47pt,62.5pt) {\footnotesize $0$};
	
\end{tikzpicture}

%% file: branch.tikz

\begin{tikzpicture}
	
	\draw[very thick, blue, arrows=-latex] (50pt,25pt) -- (95pt,2.5pt);
	\draw[very thick, blue, arrows=-latex] (50pt,-25pt) -- (95pt,-2.5pt);
	\draw[very thick, brown, arrows=-latex] (0pt,0pt) -- (45pt,22.5pt);
	\draw[very thick, brown, densely dotted] (25pt,12.5pt) -- (50pt,-25pt);
	\draw[very thick, brown, arrows=-latex] (0pt,0pt) -- (45pt,-22.5pt);
	\draw[very thick, brown, densely dotted] (25pt,-12.5pt) -- (50pt,25pt);
	
	\node[anchor=center] at (22pt,18.8pt) {\small $\delta_x$};
	\node[anchor=center] at (22pt,-17.8pt) {\small $\delta_y$};
	\node[anchor=center] at (78pt,17.7pt) {\small $1$};
	\node[anchor=center] at (78pt,-17.7pt) {\small $1$};
	
	\draw[black, fill=white] (0pt,0pt) circle (1.6ex);
	\draw[black, fill=white] (50pt,25pt) circle (1.6ex);
	\draw[black, fill=white] (50pt,-25pt) circle (1.6ex);
	\draw[black, fill=white] (100pt,0pt) circle (1.6ex);
	
	\node[anchor=center] at (0pt,0pt) {\normalsize $s$};
	\node[anchor=center] at (50pt,25pt) {\normalsize $x$};
	\node[anchor=center] at (50pt,-25pt) {\normalsize $y$};
	\node[anchor=center] at (100pt,0pt) {\normalsize $t$};
	
	\draw [fill=white] (4pt,-3pt) rectangle (13pt,-10pt);
	\node[anchor=center] at (8.5pt,-6.5pt) {\scriptsize $\infty$};
	\draw [fill=white] (104.5pt,-2pt) rectangle (110.5pt,-11pt);
	\node[anchor=center] at (107.5pt,-6.5pt) {\footnotesize $0$};
	\draw [fill=white] (54.5pt,23pt) rectangle (60.5pt,14pt);
	\node[anchor=center] at (57.5pt,18.5pt) {\footnotesize $0$};
	\draw [fill=white] (54.5pt,-27pt) rectangle (60.5pt,-36pt);
	\node[anchor=center] at (57.5pt,-31.5pt) {\footnotesize $0$};
	
\end{tikzpicture}

%% file: not.tikz

\begin{tikzpicture}
	
	\draw[very thick, blue, arrows=-latex] (70pt,0pt) -- (132.5pt,0pt);
	\draw[very thick, brown, arrows=-latex] (0pt,0pt) -- (62.5pt,0pt);
	\draw[very thick, brown, densely dotted] (35pt,0pt) -- (35pt,50pt);
	
	\node[anchor=center] at (35pt,-7pt) {\footnotesize $1$};
	\node[anchor=center] at (105pt,-7pt) {\footnotesize $1$};
	
	\draw[black, fill=white] (0pt,0pt) circle (1.6ex);
	\draw[black, fill=white] (70pt,0pt) circle (1.6ex);
	\draw[black, fill=white] (140pt,0pt) circle (1.6ex);
	\draw[black, fill=white] (35pt,50pt) circle (1.6ex);
	
	\node[anchor=center] at (0pt,0pt) {\normalsize $s$};
	\node[anchor=center] at (70pt,0pt) {\normalsize $w$};
	\node[anchor=center] at (140pt,0pt) {\normalsize $t$};
	\node[anchor=center] at (35pt,50pt) {\normalsize $v$};
	
	\draw [fill=white] (4pt,-4pt) rectangle (13pt,-11pt);
	\node[anchor=center] at (8.5pt,-7.5pt) {\scriptsize $\infty$};
	\draw [fill=white] (144.5pt,-3pt) rectangle (150.5pt,-12pt);
	\node[anchor=center] at (147.5pt,-7.5pt) {\footnotesize $0$};
	\draw [fill=white] (74.5pt,-3pt) rectangle (80.5pt,-12pt);
	\node[anchor=center] at (77.5pt,-7.5pt) {\footnotesize $0$};
	
\end{tikzpicture}

%% file: or.tikz

\begin{tikzpicture}
	
	\draw[very thick, blue, arrows=-latex] (90pt,0pt) -- (152.5pt,0pt);
	\draw[very thick, brown, arrows=-latex] (0pt,2pt) -- (82.5pt,2pt);
	\draw[very thick, brown, densely dotted] (30pt,3pt) -- (30pt,35pt);
	\draw[very thick, brown, arrows=-latex] (0pt,-2pt) -- (82.5pt,-2pt);
	\draw[very thick, brown, densely dotted] (60pt,-1pt) -- (60pt,35pt);
	\draw[thick, arrows=-stealth] (30pt,70pt) -- (30pt,50pt);
	\draw[thick, arrows=-stealth] (60pt,70pt) -- (60pt,50pt);
	
	\node[anchor=center] at (45pt,-9pt) {\footnotesize $1$};
	\node[anchor=center] at (45pt,9pt) {\footnotesize $1$};
	\node[anchor=center] at (125pt,-7pt) {\footnotesize $1$};
	
	\draw[black, fill=white] (0pt,0pt) circle (1.6ex);
	\draw[black, fill=white] (90pt,0pt) circle (1.6ex);
	\draw[black, fill=white] (160pt,0pt) circle (1.6ex);
	\draw[black, fill=white] (30pt,70pt) circle (1.6ex);
	\draw[black, fill=white] (60pt,70pt) circle (1.6ex);
	
	\node[anchor=center] at (0pt,0pt) {\normalsize $s$};
	\node[anchor=center] at (90pt,0pt) {\normalsize $w$};
	\node[anchor=center] at (160pt,0pt) {\normalsize $t$};
	\node[anchor=center] at (30pt,70pt) {\normalsize $u$};
	\node[anchor=center] at (60pt,70pt) {\normalsize $v$};
	
	\draw [fill=white] (25pt,49pt) rectangle (35pt,39pt);
	\draw [fill=white] (55pt,49pt) rectangle (65pt,39pt);
	\node[anchor=center] at (30pt,45pt) {\small $\boldsymbol{\neg}$};
	\node[anchor=center] at (60pt,45pt) {\small $\boldsymbol{\neg}$};
	\draw[black, fill=white] (30pt,35pt) circle (1.6ex);
	\draw[black, fill=white] (60pt,35pt) circle (1.6ex);
	
	\draw [fill=white] (4pt,-4pt) rectangle (13pt,-11pt);
	\node[anchor=center] at (8.5pt,-7.5pt) {\scriptsize $\infty$};
	\draw [fill=white] (164.5pt,-3pt) rectangle (170.5pt,-12pt);
	\node[anchor=center] at (167.5pt,-7.5pt) {\footnotesize $0$};
	\draw [fill=white] (94.5pt,-3pt) rectangle (100.5pt,-12pt);
	\node[anchor=center] at (97.5pt,-7.5pt) {\footnotesize $0$};
	
\end{tikzpicture}

%% file: mostRepr.tikz
\definecolor{dgray}{gray}{0.35}

\begin{tikzpicture}

	\draw[thick, arrows=-stealth] (35pt,80pt) -- (35pt,45pt);
	
	\draw[very thick, brown, arrows=-latex] (0pt,0pt) -- (103pt,0pt);
	\draw[very thick, brown, arrows=-latex] (0pt,0pt) -- (85pt,0pt) -- (85pt,80pt) -- (103pt,80pt);
	\draw[very thick, brown, densely dotted] (35pt,40pt) -- (35pt,0pt);
	
	\draw[thick, blue, arrows=-{Latex[length=1.3mm, width=0.9mm]}] (110pt,0pt) -- (142.6pt,4.8pt);
	\draw[thick, blue, arrows=-{Latex[length=1.3mm, width=0.9mm]}] (110pt,0pt) -- (142.5pt,14.5pt);
	\draw[thick, blue, arrows=-{Latex[length=1.3mm, width=0.9mm]}] (110pt,0pt) -- (142.6pt,-4.8pt);
	\draw[thick, blue, arrows=-{Latex[length=1.3mm, width=0.9mm]}] (110pt,0pt) -- (142.5pt,-14.5pt);
	\draw[thick, blue, arrows=-{Latex[length=1.3mm, width=0.9mm]}] (145pt,15pt) -- (169pt,2.6pt);
	\draw[thick, blue, arrows=-{Latex[length=1.3mm, width=0.9mm]}] (145pt,5pt) -- (167.4pt,0.9pt);
	\draw[thick, blue, arrows=-{Latex[length=1.3mm, width=0.9mm]}] (145pt,-5pt) -- (167.4pt,-0.9pt);
	\draw[thick, blue, arrows=-{Latex[length=1.3mm, width=0.9mm]}] (145pt,-15pt) -- (169pt,-2.6pt);
	
	\draw[thick, blue, arrows=-{Latex[length=1.3mm, width=0.9mm]}] (145pt,83pt) -- (158pt,81pt);
	\draw[thick, blue, arrows=-{Latex[length=1.3mm, width=0.9mm]}] (145pt,77pt) -- (158pt,79pt);
	\draw[thick, blue, arrows=-{Latex[length=1.3mm, width=0.9mm]}] (145pt,70pt) -- (158pt,68pt);
	\draw[thick, blue, arrows=-{Latex[length=1.3mm, width=0.9mm]}] (145pt,64pt) -- (158pt,66pt);
	\draw[thick, blue, arrows=-{Latex[length=1.3mm, width=0.9mm]}] (145pt,95pt) -- (158pt,93pt);
	\draw[thick, blue, arrows=-{Latex[length=1.3mm, width=0.9mm]}] (145pt,89pt) -- (158pt,91pt);
	\draw[thick, brown, arrows=-{Latex[length=1.3mm, width=0.9mm]}] (110pt,80pt) -- (130pt,80pt) -- (143pt,83pt);
	\draw[thick, brown, arrows=-{Latex[length=1.3mm, width=0.9mm]}] (130pt,80pt) -- (143pt,77pt);
	\draw[thick, brown, arrows=-{Latex[length=1.3mm, width=0.9mm]}] (120pt,80pt) -- (120pt,67pt) -- (130pt,67pt) -- (143pt,70pt);
	\draw[thick, brown, arrows=-{Latex[length=1.3mm, width=0.9mm]}] (130pt,67pt) -- (143pt,64pt);
	\draw[thick, brown, arrows=-{Latex[length=1.3mm, width=0.9mm]}] (120pt,80pt) -- (120pt,92pt) -- (130pt,92pt) -- (143pt,95pt);
	\draw[thick, brown, arrows=-{Latex[length=1.3mm, width=0.9mm]}] (130pt,92pt) -- (143pt,89pt);
	\draw[thick, brown, densely dotted] (145pt,83pt) -- (137.5pt,78.5pt);
	\draw[thick, brown, densely dotted] (145pt,77pt) -- (137.5pt,81.5pt);
	\draw[thick, brown, densely dotted] (145pt,70pt) -- (137.5pt,65.5pt);
	\draw[thick, brown, densely dotted] (145pt,64pt) -- (137.5pt,68.5pt);
	\draw[thick, brown, densely dotted] (145pt,89pt) -- (137.5pt,93.5pt);
	\draw[thick, brown, densely dotted] (145pt,95pt) -- (137.5pt,90.5pt);
	
	\node[anchor=center] at (126pt,11pt) {\tiny $1$};
	\node[anchor=center] at (130pt,3pt) {\tiny $1$};
	\node[anchor=center] at (130pt,-3pt) {\tiny $1$};
	\node[anchor=center] at (126pt,-11pt) {\tiny $1$};
	
	\node[anchor=center] at (159pt,11pt) {\tiny $1$};
	\node[anchor=center] at (157pt,3pt) {\tiny $1$};
	\node[anchor=center] at (157pt,-3pt) {\tiny $1$};
	\node[anchor=center] at (159pt,-11pt) {\tiny $1$};
	
	\node[anchor=center] at (92pt,-5pt) {\scriptsize $\infty$};
	\node[anchor=center] at (92pt,85pt) {\scriptsize $\infty$};
	
	\draw [fill=white] (-2pt,110pt) rectangle (72pt,70pt);
	\node[anchor=center] at (35pt,96pt) {\scriptsize \textit{Base construction}};
	\draw[black, fill=white] (35pt,70pt) circle (1.6ex);
	\node[anchor=center] at (35.5pt,69.5pt) {\normalsize $v_I$};
	
	\draw (110pt,35pt) rectangle (180pt,-25pt);
	\node[anchor=center] at (145pt,27pt) {\scriptsize \textit{Control group}};
	\draw[black, fill=white] (145pt,15pt) circle (0.6ex);
	\draw[black, fill=white] (145pt,5pt) circle (0.6ex);
	\draw[black, fill=white] (145pt,-5pt) circle (0.6ex);
	\draw[black, fill=white] (145pt,-15pt) circle (0.6ex);
	\draw[black, fill=white] (170pt,0pt) circle (0.6ex);
	
	\draw [dashed, gray, thick] (135pt,21pt) rectangle (153pt,-21pt);
	\draw[gray, thick] (153pt,15pt) -- (186pt,15pt);
	\draw[gray, thick] (189pt,23pt) -- (186pt,23pt) -- (186pt,7pt) -- (189pt,7pt);
	\node[anchor=center] at (212pt,20pt) {\scriptsize $m^{1/\epsilon'}\!$ new};
	\node[anchor=center] at (212pt,11pt) {\scriptsize nodes};
	
	\draw (110pt,55pt) rectangle (180pt,115pt);
	\node[anchor=center] at (145pt,108pt) {\scriptsize \textit{Generating group}};
	\draw[black, fill=white] (145pt,83pt) circle (0.5ex);
	\draw[black, fill=white] (145pt,77pt) circle (0.5ex);
	\draw[black, fill=white] (160pt,80pt) circle (0.5ex);
	\draw[black, fill=white] (145pt,70pt) circle (0.5ex);
	\draw[black, fill=white] (145pt,64pt) circle (0.5ex);
	\draw[black, fill=white] (160pt,67pt) circle (0.5ex);
	\draw[black, fill=white] (145pt,89pt) circle (0.5ex);
	\draw[black, fill=white] (145pt,95pt) circle (0.5ex);
	\draw[black, fill=white] (160pt,92pt) circle (0.5ex);
	
	\draw [dashed, gray, thick] (125pt,100pt) rectangle (165pt,60pt);
	\draw[gray, thick] (165pt,95pt) -- (186pt,95pt);
	\draw[gray, thick] (189pt,106pt) -- (186pt,106pt) -- (186pt,84pt) -- (189pt,84pt);
	\node[anchor=center] at (212pt,103pt) {\scriptsize $N^2$ distinct};
	\node[anchor=center] at (212pt,94.5pt) {\scriptsize branching};
	\node[anchor=center] at (212pt,87pt) {\scriptsize gadgets};
	
	\draw [fill=white] (30pt,49pt) rectangle (40pt,37pt);
	\node[anchor=center] at (35pt,45pt) {\small $\boldsymbol{\neg}$};
	\draw[black, fill=white] (35pt,35pt) circle (1.6ex);

	\draw[black, fill=white] (0pt,0pt) circle (1.6ex);
	\node[anchor=center] at (0pt,0pt) {\small $s$};
	\draw[black, fill=white] (110pt,0pt) circle (1.6ex);
	\node[anchor=center] at (110pt,-0.5pt) {\small $s_{_{\!}c}$};
	\draw[black, fill=white] (110pt,80pt) circle (1.6ex);
	\node[anchor=center] at (110pt,79.5pt) {\small $s_{_{\!}g}$};
	
	\draw [fill=white] (4pt,-3pt) rectangle (13pt,-10pt);
	\node[anchor=center] at (8.5pt,-6.5pt) {\scriptsize $\infty$};
	\draw [fill=white] (113.5pt,-3pt) rectangle (119.5pt,-12pt);
	\node[anchor=center] at (116.5pt,-7.5pt) {\footnotesize $0$};
	\draw [fill=white] (113.5pt,77.5pt) rectangle (119.5pt,68.5pt);
	\node[anchor=center] at (116.5pt,73pt) {\footnotesize $0$};

\end{tikzpicture}

%% file: subopt.tikz
\definecolor{dgray}{gray}{0.35}

\begin{tikzpicture}

	\draw[thick, arrows=-stealth] (75pt,80pt) -- (75pt,50pt);
	\draw[very thick, blue, arrows=-latex] (-2pt,40pt) -- (-2pt,74pt);
	\draw[very thick, blue, arrows=-latex] (2pt,80pt) -- (2pt,46pt);
	
	\draw[very thick, brown, arrows=-latex] (0pt,2pt) -- (69pt,2pt);
	\draw[very thick, brown, arrows=-latex] (0pt,-2pt) -- (69pt,-2pt);
	\draw[very thick, brown, densely dotted] (0pt,40pt) -- (0pt,25pt) -- (35pt,25pt) -- (35pt,2pt);
	\draw[very thick, brown, densely dotted] (75pt,40pt) -- (75pt,25pt) -- (40pt,25pt) -- (40pt,-2pt);
	
	\draw[thick, brown, densely dotted] (0pt,40pt) -- (20pt,40pt) -- (20pt,120pt) -- (124.5pt,120pt)  -- (124.5pt,105pt);
	\draw[thick, brown, densely dotted] (121.5pt,120pt) -- (121.5pt,90pt);
	\draw[thick, brown, densely dotted] (118.5pt,120pt) -- (118.5pt,75pt);
	\draw[thick, brown, densely dotted] (115.5pt,120pt) -- (115.5pt,40pt);
	
	\node[anchor=center] at (-6pt,60pt) {\scriptsize $1$};
	\node[anchor=center] at (6pt,60pt) {\scriptsize $1$};
	\node[anchor=center] at (25pt,9pt) {\scriptsize $1$};
	\node[anchor=center] at (50pt,-9pt) {\scriptsize $1$};
	
	\draw [fill=white] (130pt,110pt) rectangle (160pt,100pt);
	\draw [fill=white] (130pt,95pt) rectangle (160pt,85pt);
	\draw [fill=white] (130pt,80pt) rectangle (160pt,70pt);
	\draw [fill=white] (130pt,45pt) rectangle (160pt,35pt);
	
	\draw[thick, brown, arrows=-latex] (75pt,105pt) -- (140pt,105pt);
	\draw[thick, brown, arrows=-latex] (75pt,90pt) -- (140pt,90pt);
	\draw[thick, brown, arrows=-latex] (75pt,75pt) -- (140pt,75pt);
	\draw[thick, brown, arrows=-latex] (75pt,40pt) -- (140pt,40pt);
	
	\draw [fill=white] (38pt,110pt) rectangle (111.5pt,70pt);
	\node[anchor=center] at (75pt,96pt) {\scriptsize \textit{Base construction}};
	\draw[black, fill=white] (75pt,70pt) circle (1.6ex);
	\node[anchor=center] at (75.5pt,69.5pt) {\normalsize $v_I$};
	
	\draw [fill=white] (70pt,54pt) rectangle (80pt,42pt);
	\node[anchor=center] at (75pt,50pt) {\small $\boldsymbol{\neg}$};
	\draw[black, fill=white] (75pt,40pt) circle (1.6ex);
	
	\draw[gray, thick] (128pt,23pt) -- (128pt,20pt) -- (162pt,20pt) -- (162pt,23pt);
	\draw[gray, thick] (145pt,20pt) -- (145pt,17pt);
	\node[anchor=center] at (145pt,11pt) {\scriptsize unhappy};
	\node[anchor=center] at (145pt,3pt) {\scriptsize penalty};
	\node[anchor=center] at (145pt,-5pt) {\scriptsize gadgets};

	\draw[black, fill=white] (0pt,0pt) circle (1.6ex);
	\node[anchor=center] at (0pt,-0.5pt) {\small $s_0$};
	\draw[black, fill=white] (75pt,0pt) circle (1.6ex);
	\node[anchor=center] at (75pt,0pt) {\small $w$};
	\draw[black, fill=white] (0pt,40pt) circle (1.6ex);
	\node[anchor=center] at (0pt,40pt) {\small $v_0$};
	\draw[black, fill=white] (0pt,80pt) circle (1.6ex);
	\node[anchor=center] at (0pt,80.5pt) {\small $v_0'$};
	
	\draw [fill=white] (3.5pt,-3pt) rectangle (9.5pt,-12pt);
	\node[anchor=center] at (6.5pt,-7.5pt) {\scriptsize $2$};
	\draw [fill=white] (78.5pt,-3pt) rectangle (84.5pt,-12pt);
	\node[anchor=center] at (81.5pt,-7.5pt) {\footnotesize $0$};
	\draw [fill=white] (4.5pt,38pt) rectangle (10.5pt,29pt);
	\node[anchor=center] at (7.5pt,33.5pt) {\footnotesize $0$};
	\draw [fill=white] (4.5pt,78pt) rectangle (10.5pt,69pt);
	\node[anchor=center] at (7.5pt,73.5pt) {\footnotesize $0$};

\end{tikzpicture}

%% file: cycle.tikz

\begin{tikzpicture}
	
	
	\draw[very thick, blue, arrows=-latex] (0pt,-3pt) -- (64pt,-3pt);
	\draw[very thick, blue, arrows=-latex] (70pt,3pt) -- (6pt,3pt);
	
	\node[anchor=center] at (35pt,-9pt) {\footnotesize $1$};
	\node[anchor=center] at (35pt,9pt) {\footnotesize $1$};
	
	\draw[black, fill=white] (0pt,0pt) circle (1.6ex);
	\draw[black, fill=white] (70pt,0pt) circle (1.6ex);
	
	\node[anchor=center] at (0.5pt,-0.5pt) {\normalsize $v_1$};
	\node[anchor=center] at (70.5pt,-0.5pt) {\normalsize $v_2$};
	
	\draw [fill=white] (1.5pt,-4pt) rectangle (7.5pt,-13pt);
	\node[anchor=center] at (4.5pt,-8.5pt) {\footnotesize $0$};
	\draw [fill=white] (71.5pt,-4pt) rectangle (77.5pt,-13pt);
	\node[anchor=center] at (74.5pt,-8.5pt) {\footnotesize $0$};
	
\end{tikzpicture}

%% file: notation.tikz

\begin{tikzpicture}

	\draw[thick, arrows=-stealth] (0pt,120pt) -- (0pt,101pt);
	\draw[black, fill=white] (0pt,120pt) circle (1.6ex);
	\node[anchor=center] at (0pt,120pt) {\normalsize $v$};
	\draw [fill=white] (-15pt,100pt) rectangle (15pt,89pt);
	\node[anchor=center] at (0pt,95pt) {\footnotesize $\eta_1/ \, \eta_2$};
	\draw[black, fill=white] (0pt,85pt) circle (1.6ex);
	\node[anchor=center] at (0pt,85pt) {\normalsize $w$};
	\node[anchor=center] at (0pt,65pt) {\scriptsize Cutoff gadget};
	
	\draw[thick, arrows=-stealth] (60pt,120pt) -- (60pt,100pt);
	\draw[black, fill=white] (60pt,120pt) circle (1.6ex);
	\node[anchor=center] at (60pt,120pt) {\normalsize $v$};
	\draw [fill=white] (55pt,99pt) rectangle (65pt,89pt);
	\node[anchor=center] at (60pt,95pt) {\small $\boldsymbol{\neg}$};
	\draw[black, fill=white] (60pt,85pt) circle (1.6ex);
	\node[anchor=center] at (60pt,85pt) {\normalsize $w$};
	\node[anchor=center] at (60pt,65pt) {\footnotesize \textsc{not} \scriptsize gate};
	
	\draw[thick] (105pt,125pt) -- (105pt,112pt) -- (120pt,112pt);
	\draw[thick, arrows=-stealth] (135pt,125pt) -- (135pt,112pt) -- (120pt,112pt) -- (120pt,101pt);
	\draw[black, fill=white] (105pt,125pt) circle (1.6ex);
	\node[anchor=center] at (105pt,125pt) {\normalsize $u$};
	\draw[black, fill=white] (135pt,125pt) circle (1.6ex);
	\node[anchor=center] at (135pt,125pt) {\normalsize $v$};
	\draw [fill=white] (115pt,100pt) rectangle (125pt,89pt);
	\node[anchor=center] at (120pt,96pt) {\footnotesize $\boldsymbol{\vee}$};
	\draw[black, fill=white] (120pt,85pt) circle (1.6ex);
	\node[anchor=center] at (120pt,85pt) {\normalsize $w$};
	\node[anchor=center] at (120pt,65pt) {\footnotesize \textsc{or} \scriptsize gate};
	
	\draw[thick] (165pt,125pt) -- (165pt,112pt) -- (180pt,112pt);
	\draw[thick, arrows=-stealth] (195pt,125pt) -- (195pt,112pt) -- (180pt,112pt) -- (180pt,101pt);
	\draw[black, fill=white] (165pt,125pt) circle (1.6ex);
	\node[anchor=center] at (165pt,125pt) {\normalsize $u$};
	\draw[black, fill=white] (195pt,125pt) circle (1.6ex);
	\node[anchor=center] at (195pt,125pt) {\normalsize $v$};
	\draw [fill=white] (175pt,100pt) rectangle (185pt,89pt);
	\node[anchor=center] at (180pt,96pt) {\footnotesize $\boldsymbol{\wedge}$};
	\draw[black, fill=white] (180pt,85pt) circle (1.6ex);
	\node[anchor=center] at (180pt,85pt) {\normalsize $w$};
	\node[anchor=center] at (180pt,65pt) {\footnotesize \textsc{and} \scriptsize gate};
	
\end{tikzpicture}

%% file: and_alter.tikz

\begin{tikzpicture}
	
	\draw[very thick, blue, arrows=-latex] (90pt,0pt) -- (152.5pt,0pt);
	\draw[very thick, brown, arrows=-latex] (0pt,2pt) -- (82.5pt,2pt);
	\draw[very thick, brown, densely dotted] (30pt,3pt) -- (30pt,35pt);
	\draw[very thick, brown, arrows=-latex] (0pt,-2pt) -- (82.5pt,-2pt);
	\draw[very thick, brown, densely dotted] (60pt,-1pt) -- (60pt,35pt);
	\draw[thick, arrows=-stealth] (30pt,70pt) -- (30pt,50pt);
	\draw[thick, arrows=-stealth] (60pt,70pt) -- (60pt,50pt);
	\draw[thick, arrows=-stealth] (90pt,0pt) -- (90pt,-20pt);
	
	\node[anchor=center] at (45pt,-9pt) {\footnotesize $1$};
	\node[anchor=center] at (45pt,9pt) {\footnotesize $1$};
	\node[anchor=center] at (125pt,-7pt) {\footnotesize $2$};
	
	\draw[black, fill=white] (0pt,0pt) circle (1.6ex);
	\draw[black, fill=white] (90pt,0pt) circle (1.6ex);
	\draw[black, fill=white] (160pt,0pt) circle (1.6ex);
	\draw[black, fill=white] (30pt,70pt) circle (1.6ex);
	\draw[black, fill=white] (60pt,70pt) circle (1.6ex);
	
	\node[anchor=center] at (0pt,0pt) {\normalsize $s$};
	\node[anchor=center] at (90pt,0pt) {\normalsize $w_0$};
	\node[anchor=center] at (160pt,0pt) {\normalsize $t$};
	\node[anchor=center] at (30pt,70pt) {\normalsize $u$};
	\node[anchor=center] at (60pt,70pt) {\normalsize $v$};
	
	\draw [fill=white] (25pt,49pt) rectangle (35pt,39pt);
	\draw [fill=white] (55pt,49pt) rectangle (65pt,39pt);
	\node[anchor=center] at (30pt,45pt) {\small $\boldsymbol{\neg}$};
	\node[anchor=center] at (60pt,45pt) {\small $\boldsymbol{\neg}$};
	\draw[black, fill=white] (30pt,35pt) circle (1.6ex);
	\draw[black, fill=white] (60pt,35pt) circle (1.6ex);
	
	\draw [fill=white] (73pt,-22pt) rectangle (107pt,-33pt);
	\node[anchor=center] at (90pt,-27pt) {\scriptsize \textbf{0.7/0.8}};
	\draw[black, fill=white] (90pt,-37pt) circle (1.6ex);
	\node[anchor=center] at (90pt,-37pt) {\normalsize $w$};
	
	\draw [fill=white] (4pt,-4pt) rectangle (13pt,-11pt);
	\node[anchor=center] at (8.5pt,-7.5pt) {\scriptsize $\infty$};
	\draw [fill=white] (164.5pt,-3pt) rectangle (170.5pt,-12pt);
	\node[anchor=center] at (167.5pt,-7.5pt) {\footnotesize $0$};
	\draw [fill=white] (94.5pt,-3pt) rectangle (100.5pt,-12pt);
	\node[anchor=center] at (97.5pt,-7.5pt) {\footnotesize $0$};
	
\end{tikzpicture}

%% file: unhappy.tikz

\begin{tikzpicture}
	
	\draw[very thick, blue, arrows=-latex] (80pt,40pt) -- (80pt,6pt);
	\draw[very thick, blue, arrows=-latex] (80pt,0pt) -- (125.5pt,36.5pt);
	\draw[very thick, brown, arrows=-latex] (0pt,0pt) -- (74pt,0pt);
	\draw[very thick, brown, densely dotted] (40pt,55pt) -- (40pt,0pt);
	\draw[very thick, brown, arrows=-latex] (80pt,40pt) -- (124pt,40pt);
	\draw[very thick, brown, densely dotted] (40pt,55pt) -- (105pt,55pt) -- (105pt,40pt);
	
	\node[anchor=center] at (76pt,20pt) {\footnotesize $b$};
	\node[anchor=center] at (107pt,15pt) {\footnotesize $b$};
	\node[anchor=center] at (40pt,-6pt) {\footnotesize $h$};
	\node[anchor=center] at (105pt,34pt) {\footnotesize $2$};
	
	\draw[black, fill=white] (0pt,0pt) circle (1.6ex);
	\draw[black, fill=white] (80pt,0pt) circle (1.6ex);
	\draw[black, fill=white] (80pt,40pt) circle (1.6ex);
	\draw[black, fill=white] (40pt,55pt) circle (1.6ex);
	\draw[black, fill=white] (130pt,40pt) circle (1.6ex);
	
	\node[anchor=center] at (0pt,0pt) {\normalsize $v$};
	\node[anchor=center] at (80.5pt,0pt) {\normalsize $t_0$};
	\node[anchor=center] at (80pt,40pt) {\normalsize $u$};
	\node[anchor=center] at (40.5pt,54.5pt) {\normalsize $v_0$};
	\node[anchor=center] at (130.5pt,40pt) {\normalsize $t_0'$};
	
	\draw [fill=white] (82pt,37pt) rectangle (98pt,28pt);
	\node[anchor=center] at (90pt,32.5pt) {\scriptsize $b\!+\!1$};
	\draw [fill=white] (85pt,-2pt) rectangle (91pt,-11pt);
	\node[anchor=center] at (88pt,-6.5pt) {\footnotesize $1$};
	\draw [fill=white] (135pt,38pt) rectangle (141pt,29pt);
	\node[anchor=center] at (138pt,33.5pt) {\footnotesize $0$};
	
\end{tikzpicture}

%% file: diffOpt.tikz

\begin{tikzpicture}

	\draw[dashed] (-5pt,220pt) -- (438pt,220pt);
	\draw[dashed] (220pt,352pt) -- (220pt,60pt);
	

	\draw[thick, gray] (70pt,286pt) -- (103pt,286pt);
	\draw[thick, gray] (105pt,286pt) -- (107pt,286pt);
	\draw[thick, gray] (109pt,286pt) -- (112pt,286pt) -- (112pt,346pt);
	\draw[thick, gray, arrows=-stealth] (70pt,346pt) -- (185pt,346pt);
	\draw [fill=white] (185pt,350pt) rectangle (195pt,339pt);
	\node[anchor=center] at (190pt,346pt) {\footnotesize $\boldsymbol{\wedge}$};
	\draw[black, fill=white] (190pt,335pt) circle (1.6ex);
	\node[anchor=center] at (190.5pt,334.5pt) {\normalsize $u_1$};
	
	\draw[thick, gray] (70pt,256pt) -- (107pt,256pt);
	\draw[thick, gray] (109pt,256pt) -- (116pt,256pt) -- (116pt,315pt) -- (140pt,315pt);
	\draw[thick, gray] (70pt,344pt) -- (111pt,344pt);
	\draw[thick, gray, arrows=-stealth] (113pt,344pt) -- (140pt,344pt) -- (140pt,315pt)  -- (185pt,315pt);
	\draw [fill=white] (185pt,319pt) rectangle (195pt,308pt);
	\node[anchor=center] at (190pt,315pt) {\footnotesize $\boldsymbol{\wedge}$};
	\draw[black, fill=white] (190pt,304pt) circle (1.6ex);
	\node[anchor=center] at (190.5pt,303.5pt) {\normalsize $u_2$};
	
	\draw[thick, gray] (70pt,314pt) -- (104pt,314pt) -- (104pt,284pt);
	\draw[thick, gray] (70pt,284pt) -- (107pt,284pt);
	\draw[thick, gray] (109pt,284pt) -- (115pt,284pt);
	\draw[thick, gray, arrows=-stealth] (117pt,284pt) -- (185pt,284pt);
	\draw [fill=white] (185pt,288pt) rectangle (195pt,277pt);
	\node[anchor=center] at (190pt,284pt) {\footnotesize $\boldsymbol{\wedge}$};
	\draw[black, fill=white] (190pt,273pt) circle (1.6ex);
	\node[anchor=center] at (190.5pt,272.5pt) {\normalsize $u_3$};
	
	\draw[thick, gray] (70pt,316pt) -- (108pt,316pt) -- (108pt,254pt);
	\draw[thick, gray, arrows=-stealth] (70pt,254pt) -- (185pt,254pt);
	\draw [fill=white] (185pt,258pt) rectangle (195pt,247pt);
	\node[anchor=center] at (190pt,254pt) {\footnotesize $\boldsymbol{\wedge}$};
	\draw[black, fill=white] (190pt,243pt) circle (1.6ex);
	\node[anchor=center] at (190.5pt,242.5pt) {\normalsize $u_4$};
	
	\draw[very thick, blue, arrows=-latex] (70pt,345pt) -- (100pt,330pt) -- (120pt,330pt) -- (120pt,300pt) -- (134pt,300pt);
	\draw[very thick, blue] (70pt,315pt) -- (100pt,330pt);
	\draw[very thick, blue] (70pt,285pt) -- (100pt,270pt) -- (120pt,270pt) -- (120pt,300pt);
	\draw[very thick, blue] (70pt,255pt) -- (100pt,270pt);
	\draw[very thick, brown] (0pt,300pt) -- (20pt,300pt) -- (20pt,330pt) -- (40pt,330pt);
	\draw[very thick, brown] (20pt,300pt) -- (20pt,270pt) -- (40pt,270pt);
	\draw[very thick, brown, arrows=-latex] (40pt,330pt) -- (65pt,343pt);
	\draw[very thick, brown, arrows=-latex] (40pt,330pt) -- (65pt,317pt);
	\draw[very thick, brown, arrows=-latex] (40pt,270pt) -- (65pt,283pt);
	\draw[very thick, brown, arrows=-latex] (40pt,270pt) -- (65pt,257pt);
	\draw[very thick, brown, densely dotted] (70pt,315pt) -- (55pt,337.5pt);
	\draw[very thick, brown, densely dotted] (70pt,345pt) -- (55pt,322.5pt);
	\draw[very thick, brown, densely dotted] (70pt,255pt) -- (55pt,277.5pt);
	\draw[very thick, brown, densely dotted] (70pt,285pt) -- (55pt,262.5pt);
	
	\node[anchor=center] at (49pt,340pt) {\scriptsize $2$};
	\node[anchor=center] at (49pt,320pt) {\scriptsize $1$};
	\node[anchor=center] at (49pt,280pt) {\scriptsize $2$};
	\node[anchor=center] at (49pt,260pt) {\scriptsize $1$};
	\node[anchor=center] at (91pt,340pt) {\scriptsize $1$};
	\node[anchor=center] at (91pt,320pt) {\scriptsize $1$};
	\node[anchor=center] at (91pt,280pt) {\scriptsize $1$};
	\node[anchor=center] at (91pt,260pt) {\scriptsize $1$};
	
	\draw[black, fill=white] (0pt,300pt) circle (1.5ex);
	\draw[black, fill=white] (70pt,345pt) circle (1.5ex);
	\draw[black, fill=white] (70pt,315pt) circle (1.5ex);
	\draw[black, fill=white] (70pt,285pt) circle (1.5ex);
	\draw[black, fill=white] (70pt,255pt) circle (1.5ex);
	\draw[black, fill=white] (140pt,300pt) circle (1.5ex);
	
	\node[anchor=center] at (0pt,300pt) {\normalsize $s$};
	\node[anchor=center] at (140pt,300pt) {\normalsize $t$};
	
	\draw [fill=white] (4pt,297pt) rectangle (13pt,290pt);
	\node[anchor=center] at (8.5pt,293.5pt) {\scriptsize $\infty$};
	\draw [fill=white] (144.5pt,298pt) rectangle (150.5pt,289pt);
	\node[anchor=center] at (147.5pt,293.5pt) {\footnotesize $0$};
	\draw [fill=white] (74.5pt,343pt) rectangle (80.5pt,334pt);
	\node[anchor=center] at (77.5pt,338.5pt) {\footnotesize $0$};
	\draw [fill=white] (74.5pt,313pt) rectangle (80.5pt,303pt);
	\node[anchor=center] at (77.5pt,308.5pt) {\footnotesize $0$};
	\draw [fill=white] (74.5pt,283pt) rectangle (80.5pt,274pt);
	\node[anchor=center] at (77.5pt,278.5pt) {\footnotesize $0$};
	\draw [fill=white] (74.5pt,253pt) rectangle (80.5pt,243pt);
	\node[anchor=center] at (77.5pt,248.5pt) {\footnotesize $0$};
	
	
	\draw [gray, dashed, thick] (324pt,323pt) rectangle (356pt,236pt);
	\draw [gray, thick] (300pt,280pt) -- (324pt,280pt);
	\draw [gray, thick] (297pt,291pt) -- (300pt,291pt) -- (300pt,269pt) -- (297pt,269pt);
	\node[anchor=center] at (268pt,285pt) {\footnotesize set $W_1$ of};
	\node[anchor=center] at (268pt,275pt) {\footnotesize $\Theta(\sqrt{n})$ nodes};
	
	\draw[very thick, brown, arrows=-latex] (340pt,310pt) -- (404pt,283pt);
	\draw[very thick, brown, arrows=-latex] (340pt,290pt) -- (403.6pt,281pt);
	\draw[very thick, brown, arrows=-latex] (340pt,250pt) -- (404pt,277pt);
	\draw[very thick, brown, densely dotted] (371pt,340pt) -- (371pt,296pt);
	\draw[very thick, brown, densely dotted] (375pt,340pt) -- (375pt,285pt);
	\draw[very thick, brown, densely dotted] (379pt,340pt) -- (379pt,266pt);
	
	\node[anchor=center] at (367pt,294pt) {\scriptsize $1$};
	\node[anchor=center] at (371pt,281pt) {\scriptsize $1$};
	\node[anchor=center] at (378pt,261pt) {\scriptsize $1$};
	
	\draw[black, fill=white] (340pt,310pt) circle (1.5ex);
	\draw[black, fill=white] (340pt,290pt) circle (1.5ex);
	\node[anchor=center] at (340pt,268pt) {\normalsize $...$};
	\draw[black, fill=white] (340pt,250pt) circle (1.5ex);
	
	\draw[black, fill=white] (410pt,280pt) circle (1.5ex);
	\node[anchor=center] at (410pt,280pt) {\normalsize $t$};
	\draw[black, fill=white] (375pt,340pt) circle (1.5ex);
	\node[anchor=center] at (375.5pt,339.5pt) {\normalsize $u_1$};
	
	\draw [fill=white] (344.5pt,308pt) rectangle (350.5pt,299pt);
	\node[anchor=center] at (347.5pt,303.5pt) {\footnotesize $0$};
	\draw [fill=white] (344.5pt,288pt) rectangle (350.5pt,279pt);
	\node[anchor=center] at (347.5pt,283.5pt) {\footnotesize $0$};
	\draw [fill=white] (344.5pt,248pt) rectangle (350.5pt,239pt);
	\node[anchor=center] at (347.5pt,243.5pt) {\footnotesize $0$};
	\draw [fill=white] (414.5pt,278pt) rectangle (420.5pt,269pt);
	\node[anchor=center] at (417.5pt,273.5pt) {\footnotesize $0$};
	
	
	\draw[thick, arrows=-stealth] (65pt,195pt) -- (65pt,174pt);	
	\draw [gray, dashed, thick] (84pt,158pt) rectangle (116pt,71pt);
	\draw [gray, thick] (140pt,115pt) -- (116pt,115pt);
	\draw [gray, thick] (143pt,126pt) -- (140pt,126pt) -- (140pt,104pt) -- (143pt,104pt);
	\node[anchor=center] at (169pt,120pt) {\footnotesize set $W_2$ of};
	\node[anchor=center] at (169pt,110pt) {\footnotesize $\Theta(n)$ nodes};
	
	\draw[very thick, brown, arrows=-latex] (30pt,115pt) -- (95pt,143pt);
	\draw[very thick, brown, arrows=-latex] (30pt,115pt) -- (94.5pt,124pt);
	\draw[very thick, brown, arrows=-latex] (30pt,115pt) -- (95pt,87pt);
	\draw[very thick, brown, densely dotted] (69pt,160pt) -- (69pt,132pt);
	\draw[very thick, brown, densely dotted] (65pt,160pt) -- (65pt,120pt);
	\draw[very thick, brown, densely dotted] (61pt,160pt) -- (61pt,101pt);
	
	\node[anchor=center] at (73.5pt,128.5pt) {\scriptsize $1$};
	\node[anchor=center] at (69pt,116pt) {\scriptsize $1$};
	\node[anchor=center] at (62pt,96pt) {\scriptsize $1$};
	
	\draw[black, fill=white] (100pt,145pt) circle (1.5ex);
	\draw[black, fill=white] (100pt,125pt) circle (1.5ex);
	\node[anchor=center] at (100pt,103pt) {\normalsize $...$};
	\draw[black, fill=white] (100pt,85pt) circle (1.5ex);
	
	\draw[black, fill=white] (30pt,115pt) circle (1.5ex);
	\node[anchor=center] at (30pt,115pt) {\normalsize $s$};
	\draw[black, fill=white] (65pt,195pt) circle (1.5ex);
	\node[anchor=center] at (65.5pt,194.5pt) {\normalsize $u_2$};
	
	\draw [fill=white] (60pt,174pt) rectangle (70pt,164pt);
	\node[anchor=center] at (65pt,170pt) {\small $\boldsymbol{\neg}$};
	\draw[black, fill=white] (65pt,160pt) circle (1.6ex);
	
	\draw [fill=white] (104.5pt,143pt) rectangle (110.5pt,134pt);
	\node[anchor=center] at (107.5pt,138.5pt) {\footnotesize $0$};
	\draw [fill=white] (104.5pt,123pt) rectangle (110.5pt,114pt);
	\node[anchor=center] at (107.5pt,118.5pt) {\footnotesize $0$};
	\draw [fill=white] (104.5pt,83pt) rectangle (110.5pt,74pt);
	\node[anchor=center] at (107.5pt,78.5pt) {\footnotesize $0$};
	\draw [fill=white] (34pt,112pt) rectangle (43pt,105pt);
	\node[anchor=center] at (38.5pt,108.5pt) {\scriptsize $\infty$};
	
	
	\draw[thick, arrows=-stealth] (320pt,80pt) -- (335pt,80pt) -- (335pt,90pt) -- (345pt,90pt);	

	\draw[very thick, brown, arrows=-latex] (260pt,142pt) -- (404pt,142pt);
	\draw[very thick, brown, arrows=-latex] (260pt,138pt) -- (404pt,138pt);
	\draw[very thick, brown, arrows=-latex] (300pt,160pt) -- (390pt,160pt) -- (405pt,145pt);
	\draw[very thick, brown, arrows=-latex] (300pt,120pt) -- (390pt,120pt) -- (405pt,135pt);
	\draw[very thick, brown, densely dotted] (333pt,195pt) -- (333pt,160pt);
	\draw[very thick, brown, densely dotted] (337pt,195pt) -- (337pt,142pt);
	\draw[very thick, brown, densely dotted] (320pt,80pt) -- (320pt,120pt);
	\draw[very thick, brown, densely dotted] (350pt,80pt) -- (370pt,80pt) -- (370pt,138pt);
	
	\node[anchor=center] at (352pt,147.5pt) {\scriptsize $h^2$};
	\node[anchor=center] at (365pt,165pt) {\scriptsize $h^2$};
	\node[anchor=center] at (322pt,133pt) {\scriptsize $h^3$};
	\node[anchor=center] at (338pt,115pt) {\scriptsize $h$};
	
	\draw[black, fill=white] (260pt,140pt) circle (1.5ex);
	\node[anchor=center] at (260pt,140pt) {\normalsize $s$};
	\draw[black, fill=white] (410pt,140pt) circle (1.5ex);
	\node[anchor=center] at (410pt,140pt) {\normalsize $t$};
	\draw[black, fill=white] (300pt,160pt) circle (1.5ex);
	\node[anchor=center] at (300.5pt,159.5pt) {\normalsize $w_3$};
	\draw[black, fill=white] (300pt,120pt) circle (1.5ex);
	\node[anchor=center] at (300.5pt,119.5pt) {\normalsize $w_4$};
	\draw[black, fill=white] (335pt,195pt) circle (1.5ex);
	\node[anchor=center] at (335.5pt,194.5pt) {\normalsize $u_3$};
	\draw[black, fill=white] (320pt,80pt) circle (1.5ex);
	\node[anchor=center] at (320.5pt,79.5pt) {\normalsize $u_4$};
	
	\draw [fill=white] (345pt,94pt) rectangle (355pt,84pt);
	\node[anchor=center] at (350pt,90pt) {\small $\boldsymbol{\neg}$};
	\draw[black, fill=white] (350pt,80pt) circle (1.6ex);
	
	\draw [fill=white] (301.5pt,155pt) rectangle (307.5pt,146pt);
	\node[anchor=center] at (304.5pt,150.5pt) {\footnotesize $0$};
	\draw [fill=white] (301.5pt,115pt) rectangle (307.5pt,106pt);
	\node[anchor=center] at (304.5pt,110.5pt) {\footnotesize $0$};
	\draw [fill=white] (264pt,137pt) rectangle (273pt,130pt);
	\node[anchor=center] at (268.5pt,133.5pt) {\scriptsize $\infty$};
	\draw [fill=white] (414.5pt,138pt) rectangle (420.5pt,129pt);
	\node[anchor=center] at (417.5pt,133.5pt) {\footnotesize $0$};
	
\end{tikzpicture}

%% file: dependency.tikz
\definecolor{dgreen}{RGB}{0,64,0}
\definecolor{dred}{RGB}{196,0,0}

\begin{tikzpicture}

	\draw[ultra thick, gray, arrows=-angle 60] (5pt,80pt) -- (5pt,60pt);
	\draw[ultra thick, gray, arrows=-angle 60] (125pt,80pt) -- (125pt,60pt);
	\draw[ultra thick, gray, arrows=-angle 60] (245pt,80pt) -- (245pt,60pt);
	
	\draw[very thick, dgreen, arrows=-latex] (-20pt,0pt) -- (23pt,0pt);
	\draw[very thick, dgreen, arrows=-latex] (100pt,0pt) -- (143pt,0pt);
	\draw[very thick, dgreen, arrows=-latex] (220pt,0pt) -- (263pt,0pt);
	
	\draw[very thick, blue, arrows=-latex] (-20pt,100pt) -- (23pt,100pt);
	\draw[very thick, brown, arrows=-latex] (100pt,100pt) -- (143pt,100pt);
	\draw[very thick, brown, arrows=-latex] (220pt,100pt) -- (263pt,100pt);
	\draw[very thick, brown, densely dotted] (125pt,143pt) -- (125pt,100pt);
	\draw[very thick, brown, densely dotted] (245pt,143pt) -- (245pt,100pt);
	
	\draw[very thick, blue, arrows=-latex] (245pt,143pt) -- (267pt,106pt);
	\draw[very thick, dgreen, arrows=-latex] (245pt,43pt) -- (267pt,6pt);
	\draw[very thick, dred, arrows=-{open triangle 45[length=1mm, width=0.5mm]}] (125pt,43pt) -- (147pt,6pt);

	\draw[very thick, dgreen, arrows=-latex] (245pt,43pt) -- (223pt,6pt);
	\draw[very thick, dgreen, arrows=-latex] (125pt,43pt) -- (103pt,6pt);
	
	\node[anchor=center] at (245pt,94pt) {\footnotesize $\delta_0$};
	\node[anchor=center] at (272pt,127pt) {\footnotesize $\delta \geq \delta_0$};
	
	\draw[black, fill=white] (-20pt,0pt) circle (1.6ex);
	\draw[black, fill=white] (30pt,0pt) circle (1.6ex);
	\draw[black, fill=white] (100pt,0pt) circle (1.6ex);
	\draw[black, fill=white] (150pt,0pt) circle (1.6ex);
	\draw[black, fill=white] (125pt,43pt) circle (1.6ex);
	\draw[black, fill=white] (220pt,0pt) circle (1.6ex);
	\draw[black, fill=white] (270pt,0pt) circle (1.6ex);
	\draw[black, fill=white] (245pt,43pt) circle (1.6ex);
	
	\draw[black, fill=white] (-20pt,100pt) circle (1.5ex);
	\draw[black, fill=white] (30pt,100pt) circle (1.5ex);
	\draw[black, fill=white] (100pt,100pt) circle (1.5ex);
	\draw[black, fill=white] (150pt,100pt) circle (1.5ex);
	\draw[black, fill=white] (125pt,143pt) circle (1.5ex);
	\draw[black, fill=white] (220pt,100pt) circle (1.5ex);
	\draw[black, fill=white] (270pt,100pt) circle (1.5ex);
	\draw[black, fill=white] (245pt,143pt) circle (1.5ex);
	
	\node[anchor=center] at (-20pt,0pt) {\normalsize $u$};
	\node[anchor=center] at (30pt,0pt) {\normalsize $v$};
	\node[anchor=center] at (100pt,0pt) {\normalsize $u$};
	\node[anchor=center] at (150pt,0pt) {\normalsize $v$};
	\node[anchor=center] at (220pt,0pt) {\normalsize $u$};
	\node[anchor=center] at (270pt,0pt) {\normalsize $v$};
	\node[anchor=center] at (125pt,43pt) {\normalsize $w$};
	\node[anchor=center] at (245pt,43pt) {\normalsize $w$};
	
	\node[anchor=center] at (-20pt,100pt) {\normalsize $u$};
	\node[anchor=center] at (30pt,100pt) {\normalsize $v$};
	\node[anchor=center] at (100pt,100pt) {\normalsize $u$};
	\node[anchor=center] at (150pt,100pt) {\normalsize $v$};
	\node[anchor=center] at (220pt,100pt) {\normalsize $u$};
	\node[anchor=center] at (270pt,100pt) {\normalsize $v$};
	\node[anchor=center] at (125pt,143pt) {\normalsize $w$};
	\node[anchor=center] at (245pt,143pt) {\normalsize $w$};

\end{tikzpicture}

%% file: cycleBranch.tikz
\definecolor{dgray}{gray}{0.35}

\begin{tikzpicture}

	\draw[thick, arrows=-stealth] (80pt,85pt) -- (80pt,65pt);
	\draw[thick, arrows=-stealth] (-30pt,85pt) -- (-30pt,65pt);
	\draw[very thick, blue, arrows=-latex] (-2pt,140pt) -- (-2pt,174pt);
	\draw[very thick, blue, arrows=-latex] (2pt,180pt) -- (2pt,146pt);
	
	\draw[thick] (30pt,85pt) -- (30pt,35pt) -- (55pt,35pt);
	\draw[thick, arrows=-stealth] (80pt,50pt) -- (80pt,35pt) -- (55pt,35pt) -- (55pt,21pt);
	
	\node[anchor=center] at (-6pt,160pt) {\scriptsize $1$};
	\node[anchor=center] at (6pt,160pt) {\scriptsize $1$};
	
	\draw[thick, arrows=-stealth] (0pt,140pt) -- (0pt,120pt) -- (80pt,120pt)  -- (80pt,105pt);
	\draw [fill=white] (64pt,105pt) rectangle (96pt,83pt);
	\node[anchor=center] at (80pt,96pt) {\scriptsize $\boldsymbol{\frac{2}{3}\,/\,\frac{5}{6}}$};
	\draw[black, fill=white] (80pt,85pt) circle (1.6ex);
	\node[anchor=center] at (80.5pt,84.5pt) {\normalsize $w_2$};
	
	\draw[thick, arrows=-stealth] (30pt,120pt)  -- (30pt,105pt);
	\draw [fill=white] (14pt,105pt) rectangle (46pt,83pt);
	\node[anchor=center] at (30pt,96pt) {\scriptsize $\boldsymbol{\frac{1}{6}\,/\,\frac{1}{3}}$};
	\draw[black, fill=white] (30pt,85pt) circle (1.6ex);
	\node[anchor=center] at (30.5pt,84.5pt) {\normalsize $w_1$};
	
	\draw [fill=white] (75pt,64pt) rectangle (85pt,52pt);
	\node[anchor=center] at (80pt,60pt) {\small $\boldsymbol{\neg}$};
	\draw[black, fill=white] (80pt,50pt) circle (1.6ex);
	
	\draw[thick, arrows=-stealth] (0pt,140pt) -- (0pt,120pt) -- (-80pt,120pt)  -- (-80pt,105pt);
	\draw [fill=white] (-64pt,105pt) rectangle (-96pt,83pt);
	\node[anchor=center] at (-80pt,96pt) {\scriptsize $\boldsymbol{\frac{1}{3}\,/\,\frac{1}{2}}$};
	\draw[black, fill=white] (-80pt,85pt) circle (1.6ex);
	\node[anchor=center] at (-80pt,85pt) {\normalsize $x$};
	
	\draw[thick, arrows=-stealth] (-30pt,120pt)  -- (-30pt,105pt);
	\draw [fill=white] (-14pt,105pt) rectangle (-46pt,83pt);
	\node[anchor=center] at (-30pt,96pt) {\scriptsize $\boldsymbol{\frac{1}{2}\,/\,\frac{2}{3}}$};
	\draw[black, fill=white] (-30pt,85pt) circle (1.6ex);
	
	\draw [fill=white] (-25pt,64pt) rectangle (-35pt,52pt);
	\node[anchor=center] at (-30pt,60pt) {\small $\boldsymbol{\neg}$};
	\draw[black, fill=white] (-30pt,50pt) circle (1.6ex);
	\node[anchor=center] at (-30pt,50pt) {\normalsize $y$};
	
	\draw [fill=white] (50pt,20pt) rectangle (60pt,9pt);
	\node[anchor=center] at (55pt,16pt) {\footnotesize $\boldsymbol{\wedge}$};
	\draw[black, fill=white] (55pt,5pt) circle (1.6ex);
	
	\draw[very thick, gray, arrows=-latex] (30pt,5pt) -- (56pt,5pt);
	\draw[thick, gray] (27pt,13pt) -- (30pt,13pt) -- (30pt,-3pt) -- (27pt,-3pt);
	\node[anchor=center] at (0pt,9.5pt) {\scriptsize \textit{reference node}};
	\node[anchor=center] at (0pt,0.5pt) {\scriptsize \textit{for penalties}};

	\draw[black, fill=white] (0pt,140pt) circle (1.6ex);
	\node[anchor=center] at (0pt,139.5pt) {\small $v_0$};
	\draw[black, fill=white] (0pt,180pt) circle (1.6ex);
	\node[anchor=center] at (0pt,180pt) {\small $v_0'$};
	
	\draw [fill=white] (4.5pt,138pt) rectangle (10.5pt,129pt);
	\node[anchor=center] at (7.5pt,133.5pt) {\footnotesize $0$};
	\draw [fill=white] (4.5pt,178pt) rectangle (10.5pt,169pt);
	\node[anchor=center] at (7.5pt,173.5pt) {\footnotesize $0$};

\end{tikzpicture}